\documentclass[aps,prd,nofootinbib,notitlepage,unsortedaddress
]{revtex4-1}
\usepackage{enumerate}
\usepackage{dsfont}
\usepackage[textwidth=475px]{geometry}
\usepackage{amsfonts}
\usepackage{amssymb}
\usepackage[latin1]{inputenc}
\usepackage[T1]{fontenc}
\usepackage{graphicx}
\usepackage{mathtools}
\usepackage{graphicx}
\usepackage{subcaption}
\usepackage{bm}
\usepackage{bbm}
\usepackage{ulem}
\usepackage{url}

\usepackage[dvipsnames]{xcolor}

\usepackage[colorlinks=true]{hyperref}
\usepackage{amsthm}

\newcommand{\R}{\mathbb R}

\newcommand{\D}{\text{d}}

\newcommand{\db}{\bar\partial}
\newcommand{\sgn}{\text{sgn}}

\newcommand{\ab}{$(\alpha,\beta)$}

\newcommand{\note}[1]{{\color{black}{#1}}}

 \newtheorem{prop}{Proposition}
\newtheorem{cor}[prop]{Corollary}

\newtheorem{theor}[prop]{Theorem}

\newtheorem{defi}[prop]{Definition}

\newcommand{\nocontentsline}[3]{}
\newcommand{\tocless}[2]{\bgroup\let\addcontentsline=\nocontentsline#1{#2}\egroup}

\begin{document}
\title{Finsler Gravitational Waves of \ab-Type and their Observational Signature}

\author{Sjors Heefer}
\email{s.j.heefer@tue.nl}
\affiliation{Department of Mathematics and Computer Science, Eindhoven University of Technology, Eindhoven, The Netherlands}
\author{Andrea Fuster}
\email{a.fuster@tue.nl}
\affiliation{Department  of Mathematics and Computer Science, Eindhoven University of Technology, Eindhoven, The Netherlands}

\begin{abstract}
We introduce a new class of $(\alpha,\beta)$-type exact solutions in Finsler gravity closely related to the well-known pp-waves in general relativity. Our class contains most of the exact solutions currently known in the literature as special cases. The linearized versions of these solutions may be interpretted as Finslerian gravitational waves, and we investigate the physical effect of such waves. More precisely, we compute the Finslerian correction to the radar distance along an interferometer arm at the moment a Finslerian gravitational wave passes a detector. We come to the remarkable conclusion that the effect of a Finslerian gravitational wave on an interferometer is indistinguishable from that of standard gravitational wave in general relativity. Along the way we also physically motivate a modification of the Randers metric and prove that it has some very interesting properties.
\end{abstract}

\maketitle

\newpage
\tableofcontents



\newpage
\color{black}
\section{Introduction}

Even though in the general theory of relativity (GR) the geometry of spacetime is modelled by a \mbox{(pseudo-)Riemannian} metric of Lorentzian signature, 
%
%
there is no clear physical principle, nor experimental evidence, that tells us that this spacetime geometry should necessarily be \mbox{(pseudo-)Riemannian}. In fact, as suggested already in 1985 by Tavakol and Van den Bergh \cite{TAVAKOL198523, Tavakol_1986, Tavakol2009}, the axiomatic approach by Ehlers, Pirani and Schild (EPS) \cite{Ehlers2012} is compatible with Finsler geometry, a natural extension of \mbox{(pseudo-)Riemannian} geometry. This was originally overlooked due to too restrictive differentiability assumptions, as recently pointed out in \cite{Lammerzahl:2018lhw} and then worked out in detail in \cite{Bernal_2020}. Other axiomatic approaches also allow for types of geometry more general than the type used in GR, see e.g. \cite{Bubuianu2018}. This indicates that such types of geometries should not a priori be excluded from our theories and motivates the study of extensions of general relativity based on more general spacetime geometries. \\

In this regard Finsler geometry is the natural candidate as it provides the most general geometric framework that is still compatible with the clock postulate in the usual sense, namely that the proper time interval measured by an observer between two events can be defined as the length of its worldline connecting these events, in this case the Finslerian length rather than the \mbox{(pseudo-)Riemannian} length. We remark that Weyl geometry, another generalization of Lorentzian geometry, is also compatible with the clock postulate, but in that case the definition of proper time has to be revised \cite{Perlick1987}. \\ 

Further motivation for the study of Finsler spacetime geometry comes from quantum gravity phenomenology \cite{Addazi_2022}. 
Inspired by various approaches to quantum gravity, a generic feature of phenomenological or effective quantum gravity models is the presence of Planck-scale modified dispersion relations (MDR), related to departure from (local) Lorentz symmetry \cite{Addazi_2022, Amelino_Camelia_2013, Mattingly_2005}, which may manifest either in the sense of Lorentz invariance violation (LIV) or in the sense of deformed Lorentz symmetry. It turns out that such MDRs generically induce a Finsler geometry on spacetime \cite{Girelli:2006fw}. The mathematical details of this were investigated in \cite{Raetzel:2010je,Rodrigues:2022mfj}; see e.g. \cite{Amelino-Camelia:2014rga,Lobo_2017, Letizia:2016lew} for applications to specific quantum gravity phenomenology models. \\

Here we consider the (action-based) approach to Finsler gravity outlined in \cite{Pfeifer:2011xi,Hohmann_2019}. Structurally the theory is completely analogous to general relativity, but Einstein's field equation is replaced by Pfeifer and Wohlfarth's field equation. For \mbox{(pseudo-)Riemannian} spacetimes the latter reduces to the former. Although any solution to the field equations of GR is a solution in Finsler gravity, not many exact, properly Finslerian solutions are known as of yet. To the best of our knowledge the only ones currently known in the literature are the ($m$-Kropina type) Finsler pp-waves \cite{Fuster:2015tua} and their generalization as Very General Relativity (VGR)  spacetimes \cite{Fuster:2018djw}, and the Randers pp-waves \cite{Heefer_2021}.\\

Here we introduce a large class of exact vacuum solutions that contains most of the aforementioned solutions as special cases, the only exception being those solutions in \cite{Fuster:2018djw} that are not of pp-wave type. Namely, we prove that any Finsler metric constructed from a \mbox{(pseudo-)Riemannian} metric $\alpha$ and a 1-form $\beta$ that is covariantly constant with respect to $\alpha$, is an exact vacuum solution in Finsler gravity if $\alpha$ is a vacuum solution in general relativity. We classify all such solutions, leading to two possibilities: either $\alpha$ is flat Minkowski space, or $\alpha$ is a pp-wave. Our solutions are \ab-metrics of Berwald type.\\

The natural question that arises is whether and how such spacetimes can be physically distinguished from their general relativistic counterparts. To answer this question we consider the \note{linearized} versions of our exact solutions, which may be interpretted as Finslerian gravitational waves, and we study their physical effect. More precisely, we ask the question what would be observed in an interferometer experiment when such a Finslerian gravitational wave would pass the earth, and what would be the difference with a classical general relativistic gravitational wave. The relevant observable measured in interferometer experiments is essentially the radar distance, \note{so we compute this radar distance for our Finslerian gravitational waves, reproducing in the appropriate limit the radar distance formula for a standard gravitational wave in GR \cite{Rakhmanov_2009}. } Although at first sight the expression for the Finsler radar length looks different from the corresponding expression GR, we show that this is nothing but a coordinate artifact. Remarkably, when the two expressions are interpreted correctly in terms of observable quantities, it becomes clear that there is in fact no observational difference between the Finsler and GR case, at least as far as radar distance measurements are concerned. We discuss the significance of this. To the best of our knowledge this is the first time an explicit expression for the Finslerian Radar length has been obtained in the case of finite spacetime separations, and as such our work may be seen as a proof of concept. In contrast, the radar length for infinitesimal separations has been studied in \cite{Pfeifer_2014,Gurlebeck:2018nme}.  \\



We do point out that our results rely on the assumption that the amplitude of the gravitational wave, as well as the parameter $\lambda$ that characterized the departure from (pseudo)-Riemannian geometry, are sufficiently small, so that a certain perturbative expansion is valid. This nevertheless seems physically justified. We argue in a heuristic manner that up to first order in $\lambda$, any physically viable \ab-metric can be equivalently described by a slightly modified version of a standard Randers metric.\\

Indeed, the causal structure of the standard Randers metric does not in general have a straightforward physical interpretation. We therefore propose to modify the Randers metric slightly, only changing some relative signs in different subsets of the tangent bundle. We then prove that these modified Randers metrics have the nice property that their causal structure is completely equivalent to the causal structure of some auxiliary \mbox{(pseudo-)Riemannian} metric. This analysis is done in full generality, i.e. not just for our exact solutions. In the special case, however, that the defining 1-form of the Randers metric is covariantly constant (as is the case for our solutions) we prove that not only the causal structure, but also the affine structure of the Finsler metric coincides with that of the auxilliary (pseudo)-Riemannian metric, i.e. the timelike, spacelike and null geodesics of the Finsler metric can be understood, respectively, as the timelike, spacelike and null geodesics of the auxiliary (pseudo)-Riemannian metric. This leads to the particularly nice property that the existence of radar neighborhoods is guaranteed \cite{Perlick2008}, i.e. that given an observer and any event in spacetime, there is (at least locally) exactly one future pointing light ray and one past pointing light ray that connect the event to the worldline of the observer. This is of essential importance in our work, because without this property the notion of radar distance would not even make sense.

\subsection{Structure of this article}

The paper is organized as follows. We begin in Section \ref{sec:Finsler_gravity} with a discussion of Finsler geometry and the core ideas behind Finsler gravity. Then in Section \ref{sec:ab_metrics} we introduce $(\alpha,\beta)$-metrics, and in particular Randers metrics and discuss their relevance to Finsler gravity. We then introduce our new  solutions to the field equations and show that after linearization these solutions may be interpretted as Finslerian gravitational waves. Next, in section \ref{sec:Randers} we propose our modification of the standard Randers metric and prove that it has very satisfactory properties with respect to its causal structure, affine structure, Lorentzian signature, etc. \note{Section \ref{sec:radar_distance} is devoted to the calculation of the radar distance at the moment a Finsler gravitational wave passes, say, the Earth. We clearly point out the differences with the general relativity case.} We conclude in section \ref{sec:discussion}.


\section{Finsler Gravity}
\label{sec:Finsler_gravity}


\note{In this section we recall the basic definitions in Finsler geometry and the basic ideas that underlie Finsler gravity. We will be brief and to the point; for a slightly more detailed introduction we refer the reader to our previous article \cite{Heefer_2021}. We start with some notational remarks. Throughout this article will use induced local coordinates on the tangent bundle of a smooth manifold that can be introduced in the following way. Given a chart $\phi:U\subset M\to \R^n$ on a smooth manifold $M$, we identify any $p\in U$ with its image $\phi(p) = (x^0,\dots, x^{n-1})\in\phi(U)\subset\R^n$ under $\phi$. For $p\in U$ and $Y_p\in T_pM$, where $T_pM$ is the tangent space to $M$ at $p$, we can express $Y_p = y^\mu\partial_\mu|_p$ in terms of the holonomic basis $\partial_\mu \equiv \partial/\partial x^\mu$ of $T_pM$. This decomposition induces local coordinates $(x,y)\equiv (x^0,\dots, x^{n-1}, y^0, \dots, y^{n-1})\in\phi(U)\times\R^{n}$ on the tangent bundle $TM$ (and any open submanifold thereof). We will thus generally represent any point $(p,Y_p)\in TM$ by the tuple $(x,y)$.
The holonomic basis vectors, corresponding to these coordinates, of the tangent space $T_{(x,y)}TM$ to $TM$ at $(x,y)$ will be denoted by $\partial_\mu=\partial/\partial x^\mu$ and $\bar{\partial}_\mu=\partial/\partial y^\mu$. Throughout the article we will also assume that the spacetime dimension is $1+3$, and we will use the signature convention $(-,+,+,+)$.}


\subsection{Finsler spacetime geometry}\label{sec:FinslerSpacetimes}

 \note{A Finsler spacetime is a triple $(M,\mathcal A,F)$, where $M$ is a smooth manifold, $\mathcal A$ is a conic subbundle\footnote{\note{By a conic subbundle with non-empty fibers we mean an open subset $\mathcal A\subset TM\setminus 0$ such that $(x,\lambda y)\in\mathcal A$ for any $(x,y)\in\mathcal A$ and any $\lambda>0$, and such that $\pi(\mathcal A) = M$, where $\pi:TM\to M$ is the canonical projection of the tangent bundle.}} of $TM\setminus 0$ with non-empty fibers, and $F$, the so-called Finsler metric, is a map $F:\mathcal A\to \R$ that satisfies the following axioms:
\begin{itemize}
	\item $F$ is (positively) homogeneous of degree one with respect to $y$:
	\begin{align}
	F(x,\lambda y) =\lambda F(x, y)\,,\quad \forall \lambda>0\,;
	\end{align}
	\item The \textit{fundamental tensor}, with components $g_{\mu\nu} = \db_\mu\db_\nu \left(\frac{1}{2}F^2\right)$, has Lorentzian signature on $\mathcal{A}$.
\end{itemize}

The definition of a Finsler spacetime given above is a very \textit{weak} one in the sense that most other definitions appearing in the literature are more restrictive (see e.g. \cite{Beem, Asanov, Pfeifer:2011tk, Pfeifer:2011xi,Lammerzahl:2012kw, Javaloyes2014-1, Javaloyes2014-2}). Accordingly, our definition allows for a lot of instances, many of which will not be physically viable. This is, in our opinion, a feature rather than a bug, as most of the results in this article can be proven without further restrictions. It should be understood, however, that in order to guarantee that a viable physical interpretation is possible, the geometry should be subjected to more stringent requirements.}\\

\note{Given a Finsler metric $F$, the length of a curve $\gamma:\lambda\mapsto \gamma(\lambda)$ on $M$ can be defined as
\begin{align}
L(\gamma)=\int F(\dot{\gamma})\,\D \lambda = \int  F(x,\dot{x})\,\D \lambda,\qquad \dot{\gamma}=\frac{d\gamma}{d\lambda},
\end{align}
which, due to homogeneity, is invariant under (orientation-preserving) reparameterization. %
%
It follows by Euler's theorem for homogeneous functions that 
\begin{align}
g_{\mu\nu}(x,y)y^\mu y^\nu = F(x,y)^2,
\end{align}
and hence the length of curves is formally identical to the length in Riemannian geometry, the difference being that now the `metric tensor' may depend not only on position $x$ but on the direction $y$ as well. In fact if $g_{\mu\nu} = g_{\mu\nu}(x)$, or, equivalently, if $F^2$ is quadratic in $y$, then Finsler spacetime geometry reduces to classical Lorentzian geometry.\\

The fundamental theorem of Riemannian geometry generalizes to what's sometimes called the fundamental lemma of Finsler geometry. It states that any Finsler metric admits a unique homogeneous (nonlinear) connection on the subbundle $\mathcal A\subset TM\setminus 0$, characterized by its connection coefficents $N^\rho_\mu$, that is torsion-free, $\bar\partial_\nu N^\rho_\mu = \bar\partial_\mu N^\rho_\nu$, and metric-compatible, $\delta_\mu F^2 =0$, where $\delta_\mu \equiv \partial_\mu-N^\rho_\mu\bar\partial_\rho$ is the \textit{horizontal derivative} induced by the connection. This connection is usually referred to as the Cartan nonlinear connection or the canonical nonlinear connection and its connection coefficients are given by
\begin{align}
N^\rho_\mu(x,y) = \frac{1}{4}\bar{\partial}_\mu \bigg(g^{\rho \sigma}\big(y^\nu\partial_\nu\bar{\partial}_\sigma F^2 - \partial_\sigma F^2\big)\bigg)\,
\end{align}
where $g^{\rho \sigma}$ is the matrix inverse of the fundamental tensor $g_{\mu\nu}$. Parallel transport of a vector $V$ along a curve $\gamma$ is then characterized by the parallel transport equation\footnote{Note that the parallel transport map is in general nonlinear. Some authors (e.g. \cite{Bucataru}) choose to define parallel transport differently, namely by requiring a priori that parallel transport should be linear, which leads to the alternative parallel transport equation $\dot V^i + N^i_j(\gamma,\dot\gamma) V^i=0$. This approach, however, seems unnatural to us. Here we follow e.g. \cite{Szilasi}, where parallel transport of a vector is defined via its unique horizontal lift along a given curve. In this case parallel transport is linear if and only if the connection is linear.}
\begin{align}
\label{eq:nonlinear.parallel.transport.eq}
\dot V^\mu + N^\mu_\nu(\gamma,V)\dot \gamma^\mu = 0\,,
\end{align}
and consequently, autoparallels are those curves that satisfy
\begin{align}
\label{eq:nonlinear.geodesic.eq}
\ddot \gamma^\mu + N^\mu_\nu(\gamma,\dot \gamma)\dot \gamma^\nu = 0\,.
\end{align}
The curvature tensor, Finsler Ricci scalar and the Finsler Ricci tensor 
of $(M,F)$ are defined, respectively, as
\begin{align}\label{eq:definition_curvatures}
R^\rho{}_{\mu\nu}(x,y) =  \delta_\mu N^\rho_\nu(x,y)-\delta_\nu N^\rho_\mu(x,y),\quad \text{Ric}(x,y) = R^\rho{}_{\rho\mu}(x,y)y^\mu,\quad R_{\mu\nu}(x,y) = \frac{1}{2}\db_\mu \db_\nu\text{Ric}.
\end{align}

\subsection{Berwald spacetimes}\label{sec:Berwald}

A Berwald spacetime is a Finsler spacetime for which the Cartan nonlinear connection reduces to a \textit{linear} connection\footnote{See \cite{Szilasi2011} for an overview of the various equivalent characterizations of Berwald spaces and \cite{Pfeifer:2019tyy} for a more recent equivalent characterization.},  which is the case if and only if the connection coefficients are of the form
\begin{align}
N^\rho_\mu(x,y) = \Gamma^\rho_{\mu\nu}(x)y^\nu
\end{align}
for a set of functions $\Gamma^\rho_{\mu\nu}:M\to\R$. If so, the functions $\Gamma^\rho_{\mu\nu}$ can be identified as the Christoffel symbols of a torsion-free affine connection on $M$. We will refer to this affine connection as the associated affine connection, or simply \textit{the} affine connection on the Berwald spacetime. 
Since any (pseudo)-Riemannian spacetime is of Berwald type (with $\Gamma^\rho_{\mu\nu}$ given by the Levi-Civita connection), we have the following inclusions:
\begin{align*}
   &\text{(pseudo-)Riemannian} \subset \text{Berwald} \subset \text{Finsler}. 
\end{align*}
The parallel transport \eqref{eq:nonlinear.parallel.transport.eq} and autoparallel equations \eqref{eq:nonlinear.geodesic.eq} on a Berwald space reduce to the familiar equations
\begin{align}
\dot V^i + \Gamma^i_{jk}(\gamma)\dot \gamma^j V^k = 0, \qquad \ddot \gamma^i + \Gamma^i_{jk}(\gamma)\dot \gamma^j \dot \gamma^k = 0
\end{align}
in terms of the Christoffel symbols. The curvature tensors \eqref{eq:definition_curvatures} of a Berwald space can be written as
\begin{align}
\label{eq:symm_ricci}
R^j{}_{kl} = \bar R_i{}^j{}_{kl}(x)y^i, \qquad \text{Ric} = \bar R_{ij}(x)y^i y^j, \qquad R_{ij} = \frac{1}{2}\left(\bar R_{ij}(x) + \bar R_{ji}(x)\right),
\end{align}
in terms of the Riemann tensor $\bar R_l{}^i{}_{jk}= 2\partial_{[j} \Gamma^i_{k]l} + 2\Gamma^i_{m[j}\Gamma^m_{k]l}$ and Ricci tensor $\bar R_{lk} = \bar R_l{}^i{}_{ik}$ of the associated affine connection, where we have used the notation $T_{[ij]} = \frac{1}{2}\left(T_{ij}-T_{ji}\right)$ and $T_{(ij)} = \frac{1}{2}\left(T_{ij}+T_{ji}\right)$ for \mbox{(anti-)}symmetrization. In fact, for positive definite Finsler spaces, it follows by Szabo's metrization theorem that $R_{ij} = \frac{1}{2}\left(\bar R_{ij} + \bar R_{ji}\right) = \bar R_{ij} $, but this does not extend to Finsler spacetimes in general \cite{Fuster_2020}.}

\subsection{A note about causal structure and physical interpretation}
\label{sec:Finsler_spacetimes_interpretation}

Given a Finsler spacetime geometry, it is natural to postulate, in analogy with GR, that matter travels along timelike geodesics and light travels on null geodesics. The generalization of the notion of \textit{null} direction is mathematically straightforward. A vector $y^u$ at a point $x^\mu$ is said to be null (or lightlike)  if $F(x,y)^2 = g_{\mu\nu}(x,y)y^\mu y^\nu=0$. However, the structure of the light cone, composed of such null vectors, may be non-trivial. In GR it is always the case that the light cone separates the tangent space at each point into three connected components, that we may interpret as forward-pointing timelike vectors, backward-pointing timelike vectors, and spacelike vectors, respectively. It is then a consequence that a timelike vector is one that has positive (or negative, depending on the convention) Riemannian norm. For a generic Finsler spacetime geometry these properties of the lightcone structure are by no means guaranteed and as such it is not obvious in general how to even define what one means a by timelike vector. It certainly does not suffice to define them as positive length vectors. We do not discuss this issue any further in its full generality here. Only in the specific case of the Randers metric, in Section \ref{sec:Randers}, will we dive into the details. We argue that the causal structure of the standard Randers metric does not have a straightforward physical interpretation, but we prove that, by modifying the definition only slightly, the causal structure of such a modified Randers metric has exactly the desirable properties mentioned above in the case of GR, allowing for a straightforward physical interpretation. This will be exploited in Section \ref{sec:radar_distance}, where we compute the radar distance for a Finslerian gravitational wave of (modified) Randers type passing an interferometer.\\

It is worth mentioning that in the ideal case the (forward and backward) timelike cones should be contained in the subbundle $\mathcal A$. This statement is essentially the condition that geometry is well-defined for all timelike infinitesimal spacetime separations. This property is satisfied by our modified Randers metrics (up to a set of measure zero). It can be argued that it is not strictly necessary for spacelike vectors to be contained in $\mathcal A$, as it would not be possible, not even in principle, to perform any physical experiment that probes such directions. Whether the lightcone should be contained in $\mathcal A$ is a more delicate question, which we will not further explore here.


\subsection{The field equations}

In the context of Finsler gravity, arguably the simplest and cleanest proposal for a vacuum field equation was the one by Rutz \cite{Rutz}. The Rutz equation,  Ric $= 0$, can be derived from the geodesic deviation equation in complete analogy to the way Einstein's vacuum field equation, $R_{\mu\nu}=0$ (to which it reduces in the classical \mbox{(pseudo-)Riemannian} setting), can be derived by considering geodesic deviation. \\

However, it turns out that Rutz's equation is \textit{not} variational, i.e. it cannot be obtained by extremizing an action functional. In fact, its variational completion (i.e. the variational equation that is \textit{as close as possible to it}, in a well-defined sense \cite{Voicu_2015}) turns out to be the field equation that was proposed by Pfeifer and Wohlfarth in \cite{Pfeifer:2011xi} using a Finsler extension of the Einstein-Hilbert action \cite{Hohmann_2019}. This is again in complete analogy to the situation in GR, where the vacuum Einstein equation in the form $R_{\mu\nu} - \frac{1}{2}g_{\mu\nu}R = 0$ is also precisely the variational completion of the equation $R_{\mu\nu}=0$ \cite{Voicu_2015}. While in the GR case the completed equation happens to be equivalent to the former, this is not true any longer in the Finsler setting.\\

Although several other proposals have been made as well \cite{Horvath1950,Horvath1952, Ikeda1981, Asanov1983, Chang:2009pa,Kouretsis:2008ha,Stavrinos2014,Voicu:2009wi,Minguzzi:2014fxa}, we consider the Pfeifer-Wohlfarth equation\footnote{In the positive definite setting a similar field equation has been obtained by Chen and Shen \cite{Chen-Shen}.} \cite{Pfeifer:2011xi} to be by far the most promising, and from here onwards we will refer to it simply as \textit{the} vacuum field equation in Finsler gravity. We do not show the field equation in full generality here, as its general form is not required for our present purposes. In the case of Berwald spacetimes it can be expressed relatively simply as \cite{Fuster:2018djw}
\begin{align}\label{eq:BEFEs}
	\left(  F^2 g^{\mu\nu} - 3 y^\mu y^\nu \right) R_{\mu\nu}  = 0\,,
\end{align}
where $R_{\mu\nu}$ is the Finsler Ricci tensor and since we are in a Berwald setting, $R_{\mu\nu} = R_{\mu\nu}(x)$ only depends on $x$. Clearly the vanishing of the Finsler Ricci tensor is a sufficient condition for a Berwald spacetime to be a solution to Eq.\,\eqref{eq:BEFEs}. \note{In many cases of interest (but not always) it is a necessary condition as well. For instance, for Randers metrics of Berwald type it is known that the field equation \eqref{eq:BEFEs} is equivalent to  $R_{\mu\nu}=0$ \cite{Heefer_2021}, and similar results have been obtained for Finsler metrics satisfying strict smoothness requirements  \cite{javaloyes2023einsteinhilbertpalatini}. There exist Finsler metrics, however, that have a non-vanishing Finsler Ricci tensor, yet for which \eqref{eq:BEFEs} holds. An explicit example illustrating this will be provided in forthcoming work.}



\section{$(\alpha,\beta)$-Metrics}
\label{sec:ab_metrics}

\subsection{$(\alpha,\beta)$-metrics -- basic definitions}
An important class of Finsler geometries is given by the so-called $(\alpha,\beta)$-metrics. Here 
$\alpha = \sqrt{|a_{\mu\nu}\dot x^\mu\dot x^\nu|}$ and $\beta = b_\mu\dot x^\nu$ are scalar variables defined in terms of a \mbox{(pseudo-)Riemannian} metric $a_{\mu\nu}$ on $M$ and a 1-form $b_\mu$ on $M$, and an $(\alpha,\beta)$-metric is simply a Finsler metric that is constructed only from $\alpha$ and $\beta$, i.e. $F = F(\alpha, \beta)$. Due to homogeneity it follows that any such $F$ can be written in the standard form $F = \alpha\phi(\beta/\alpha)$ for some function $\phi$, at least whenever $\alpha\neq 0$. Well-known examples of $(\alpha,\beta)$-metrics are:
\begin{itemize}
\item Pseudo-Riemannian Finsler metrics $F = \alpha$;
\item Randers metrics $F = \alpha + \beta$;
\item Kropina metrics $F = \frac{\alpha^2}{\beta}$;
\item \note{$m$-Kropina metrics $F = \alpha^{1+m}\beta^{-m}$ with $m$ some real number. Also referred to as generalized Kropina metrics, Bogoslovsky metrics or Bogoslovsky-Kropina metrics.}
\end{itemize}

\note{For each of these types of \ab-metrics certain conditions need to be fulfilled in order to satisfy the definition of a Finsler spacetime \cite{Voicu23_Finsler_ab_spacetime_condition}.}

\subsection{Exact $(\alpha,\beta)$-metric solutions in Finsler gravity}
\label{sec:exact_ab_sols}

From the physical viewpoint, $(\alpha,\beta)$-metrics allow us to deform a GR spacetime $\alpha$ into a Finsler spacetime by the 1-form $\beta$. And it turns out, as we will prove below, that these types of metrics can be used to generalize some of the vacuum solutions to Einstein's field equations to properly Finslerian vacuum solutions in Finsler gravity. This procedure is possible whenever such a solution admits a covariantly constant vector field, or equivalently, 1-form. Namely: if the Lorentzian metric $\alpha$ solves the classical Einstein equations and the 1-form $\beta$ is covariantly constant with respect to $\alpha$ then any $(\alpha,\beta)$-metric constructed from the given $\alpha$ and $\beta$ is a solution to the Finslerian field equations. To see why this is true, we first recall the following well-known result (see e.g. section 6.3.2. in \cite{handbook_Finsler_vol2}): 

\begin{prop}
\label{prop:coinciding_spray}
Let $F$ be an $(\alpha,\beta)$-metric. If $\beta$ is covariantly constant with respect to $\alpha$ then $F$ is of Berwald type and the affine connection of $F$ coincides with the Levi-Civita connection of $\alpha$.
\end{prop}

If the affine connection of $F$ is the same as the connection of $\alpha$, the associated curvature tensors and (affine) Ricci tensors are also the same. So if $\alpha$ happens be a vacuum solution to Einstein gravity, i.e. its Ricci tensor vanishes, then it follows that the affine Ricci tensor of $F$ vanishes as well, which implies, by eq. \eqref{eq:BEFEs}, that $F$ is a vacuum solution to Pfeifer and Wohlfarth's field equation in Finsler gravity. We may summarize this result in the following theorem.

\begin{theor}\label{theor:(alpha,beta)solutions}
Let $F$ be any $(\alpha,\beta)$-metric such that $\alpha$ solves the classical vacuum Einstein equations and $\beta$ is covariantly constant with respect to $\alpha$. Then $F$ is a vacuum solution to the field equation in Finsler gravity.
\end{theor}

 In this way $(\alpha,\beta)$-metrics provide a mechanism to \textit{Finslerize} any vacuum solution to Einstein's field equations, as long as the solution admits a covariantly 1-form, or equivalently a covariantly constant vector field. The theorem generalizes some of the results obtained in \cite{Heefer_2021} for Randers metrics and in \cite{Fuster:2015tua,Fuster:2018djw} for $m$-Kropina metrics (i.e. VGR spacetimes) to arbitrary Finsler spacetimes with $(\alpha,\beta)$-metric. In particular, all pp-wave type solutions in Finsler gravity currently known in the literature are of this type.\\

Let's investigate this type of solution in some more detail. It turns out that if a vacuum solution $\alpha$ to Einstein's field equations admits a covariantly constant 1-form $\beta$, then either $\alpha$ is flat, or $\beta$ is necessarily null \cite{EhlersKundt1960} (see also \cite{Hall_2000,Batista_2014}). We remark that this result assumes that the spacetime dimension is $1+3$ and generally is not true in higher dimensions. This leads to two classes of solutions.\\

\noindent\textbf{First class of solutions}\\
The first of these possibilities, where $\alpha$ is flat, leads to a class of solutions that can always be written in suitable coordinates in the following way.\\

\fbox{\begin{minipage}{0.9\textwidth}
\textbf{$\bm{(\alpha,\beta)}$-metric solutions (Class 1).} Let the metric $A$ and 1-form $\beta$ be given by
\begin{align}\label{eq:class1_ab_sols}
A = -(\D x^0)^2 + (\D x^1)^2 + (\D x^2)^2 + (\D x^3)^2 , \qquad \beta = b_\mu \D x^\mu,
\end{align}
where $b_\mu=$ const. Then any \ab-metric constructed from $\alpha=\sqrt{|A|}$ and $\beta$ is a vacuum solution to the field equations in Finsler gravity. The resulting geometry is of Berwald type with all affine connection coefficients vanishing identically in these coordinates.
\end{minipage}}\\\\

Right below Eq. \eqref{eq:class1_ab_sols} we have used the notation $\alpha = \sqrt{|A|} = \sqrt{|a_{ij}\D x^i \D x^j|}$. This should be understood pointwise, i.e. 
\begin{align}
    \alpha = \alpha(y) = \sqrt{|a_{ij}\D x^i \D x^j|}(y) = \sqrt{|a_{ij}\D x^i(y) \D x^j(y)|} = \sqrt{|a_{ij} y^i y^j|}.
\end{align}
In other words, we sometimes write $\alpha$ for the function $\sqrt{|a_{ij}\D x^i \D x^j|}:y\mapsto \sqrt{|a_{ij} y^i y^j|}$, and at other times we write $\alpha$ for its value $\sqrt{|a_{ij} y^i y^j|}$ at $y$. It should always be clear from context what is meant.\\

\noindent\textbf{Second class of solutions}\\
The second possibility, that $\beta$ is null, leads to a class of solutions that seems to be  more interesting. In this case $\alpha$ is CCNV spacetime metric, meaning that it admits a covariantly constant null vector (CCNV), namely in this case $\beta$, or rather its vector equivalent via the isomorphism induced by $\alpha$. CCNV metrics are also known as \textit{pp-waves} (plane-fronted gravitational waves with parallel rays) and have been studied in detail in \cite{EhlersKundt1960,ehlers1962exact} (see section 24.5 in \cite{stephani_kramer_maccallum_hoenselaers_herlt_2003} for a summary). \\

It is an elementary result that by choosing suitable coordinates $(u,v,x^1,x^2)$, such $\alpha$ and $\beta$ can always be expressed in the form

\begin{align}
A &= -2\D u \left(\D v + H(u,x)\, \D u + \,W_a(u,x)\,\D x^a\right) +h_{ab}(u,x) \D x^a \D x^b, \label{eq:original_pp_wave_metric} \\
\beta &= \D u,\label{eq:original_1form}
\end{align}
where $x^a=x^1,x^2$ and $h_{ab}$ is a two-dimensional Riemannian metric. This holds irrespective of whether $\alpha$ is a solution to Einstein's field equations or not. If $\alpha$ is additionally assumed to be a vacuum solution, as in Theorem \ref{theor:(alpha,beta)solutions}, it turns out that the expression \eqref{eq:original_pp_wave_metric} for $A$ can be simplified even more \textit{without changing the form \eqref{eq:original_1form} of $\beta$}. To see this, we first consider only the metric $A$. Since $A$ is a vacuum solution to Einstein's field equations, it follows that the functions $W_a$ can be eliminated and $h_{ab}$ may be chosen as $\delta_{ab}$, by a suitable coordinate transformation (section 24.5 in \cite{stephani_kramer_maccallum_hoenselaers_herlt_2003}). The metric then takes the form
\begin{align}
A = -2\D u \left(\D v + H(u,x)\, \D u\right) +\delta_{ab} \D x^a \D x^b. \label{eq:reduced_pp_wave_metric}
\end{align}
We are, however, not only interested in the transformation behaviour of $A$ alone, but also in that of $\beta$, because an \ab-metric is composed of both. To see why we may assume without loss of generality that the form of $\beta = \D u$ remains invariant we use the fact that any coordinate transformation 
\begin{align}
(u,v,x^1,x^2)\mapsto (\bar u,\bar v,\bar x^1,\bar x^2)
\end{align}
that leaves the generic form of the metric \eqref{eq:original_pp_wave_metric} invariant, but in general changing the expressions for the metric functions $H, W_a, h_{ab}\mapsto \bar H, \bar W_a, \bar h_{ab}$, has the specific property that $u = \phi(\bar u)$ for some function $\phi$ depending on $\bar u$ alone (see section 31.2 in \cite{stephani_kramer_maccallum_hoenselaers_herlt_2003}). This applies in particular to the transformation that relates \eqref{eq:original_pp_wave_metric} and \eqref{eq:reduced_pp_wave_metric}. We can therefore express the 1-form as $\beta = \D  u = \phi'(\bar u)\D \bar u$, or equivalently $\bar b_\mu = \phi'(\bar u)\delta_\mu^u$. However, since $\beta$ is covariantly constant with respect to $A$, we must have $\bar \nabla_\mu\bar b_\nu=0$.  All  Christoffel symbols $\bar \Gamma^u_{\mu\nu}$ of the metric \eqref{eq:reduced_pp_wave_metric} with upper index $u$ vanish identically, however. Hence
\begin{align}
\bar \nabla_{\bar u} \bar b_{\bar u} = \partial_{\bar u} b_{\bar u} - \bar \Gamma^u_{uu}\phi'(\bar u) = \phi''(\bar u)\stackrel{!}{=}0.
\end{align}
It follows that $\phi'(\bar u)= C =$ constant, i.e. $\beta = C \D \bar u$. In this case it is easily seen that scaling $\bar u$ by $C$ and scaling $\bar v$ by $1/C$ leaves the metric \eqref{eq:reduced_pp_wave_metric} invariant and brings the 1-form back into its original form, proving that we may assume without loss of generality that the 1-form remains invariant under the coordinate transformation.\\

Finally, the metric \eqref{eq:reduced_pp_wave_metric} is a vacuum solution to Einstein's field equations if and only if  $(\partial_{x^1}^2+\partial_{x^2}^2)H = 0$. We may therefore characterize the second class of solutions in the following way.\\

\fbox{\begin{minipage}{0.9\textwidth}
\textbf{$\bm{(\alpha,\beta)}$-metric solutions (Class 2).} Let $\alpha = \sqrt{|A|}$ and $\beta$ be given by 
\begin{align}
A &= -2\D u \left(\D v + H(u,x)\, \D u\right) +\delta_{ab} \D x^a \D x^b, \label{eq:reduced_pp_wave_metric2} \\
\beta &= \D u \label{eq:reduced_pp_wave_1_form2},
\end{align}
such that $\delta^{ab}\partial_a\partial_b H = 0$. Then any \ab-metric constructed from the pair ($\alpha,\beta$) is a vacuum solution to the field equations in Finsler gravity. The resulting geometry is of Berwald type with affine connection identical to the Levi-Civita connection of $\alpha$.
\end{minipage}}\\\\

Note that when $H=0$ the geometries in \textit{Class 2} are also contained in \textit{Class 1}. It is not the case, however, that \textit{Class 1} is a subset of \textit{Class 2} because in \textit{Class 1} the 1-form $\beta$ need not be null,  necessarily. The preceding line of argument shows that these two classes of solutions in fact exhaust all possibilities, which we encapsulate in the following theorem.

\begin{theor}\label{theor:(alpha,beta)solutions2}
\textit{Any} vacuum solution of the type of Theorem \ref{theor:(alpha,beta)solutions} must belong to one of the two classes introduced above. 
\end{theor}

Before we move on to \ab-type solutions of plane-wave type, we end this section by noting that for specific types of $(\alpha,\beta)$-metrics, stronger results have been obtained than the ones derived above:
\begin{itemize}
\item For Randers metrics of Berwald type \textit{any} vacuum solution to \eqref{eq:BEFEs} must be of the type described in theorem \ref{theor:(alpha,beta)solutions}, that is, $\alpha$ is necessarily a vacuum solution in Einstein gravity and $\beta$ is necessarily covariantly constant \cite{Heefer_2021}. Any such solution is therefore either of \textit{Class 1} or \textit{Class 2} in the terminology introduced above.
\item For $m$-Kropina metrics some vacuum solutions of a more general type than the one in theorem \ref{theor:(alpha,beta)solutions} have been obtained in the context of \textit{Very General Relativity (VGR)} \cite{Fuster:2018djw}. 
\item Any pseudo-Riemannian Finsler metric $F = \alpha$ is trivially a vacuum solution in Finsler gravity if and only if it is a vacuum solution in Einstein gravity.
\end{itemize}

To the best of our knowledge this list comprises all exact solutions in Finsler gravity currently known in the literature.

\subsection{Plane wave solutions in Brinkman and Rosen coordinates}

Eq. \eqref{eq:reduced_pp_wave_metric2} expresses the pp-wave metric in Brinkmann form \cite{Brinkmann:1925fr}. For the description of the physical effects of (plane) gravitational waves in general relativity, it is sometimes more convenient to use a different coordinate system, known as Rosen coordinates \cite{rosen1937plane}. This remains true in the Finsler case. When we compute the effect on the radar distance of a passing Randers gravitational wave in section \ref{sec:radar_distance}, our starting point will be the expression for the gravitational wave in Rosen coordinates. Therefore we briefly review the relation between the two coordinate systems here.\\

Rosen coordinates can be introduced for the subclass of pp-waves known as \textit{plane waves}. These can be characterized by the property that the curvature tensor does not change (i.e. is covariantly constant) along the Euclidean `wave surfaces' given in Brinkmann coordinates by $\D u = \D v = 0$, i.e.
\begin{align}\label{eq:invariance_of_Riemann_curvature}
\nabla_{\partial_{x^1}} R^\rho{}_{\sigma\mu\nu} = \nabla_{\partial_{x^2}} R^\rho{}_{\sigma\mu\nu} = 0.
\end{align}
We note that $\nabla_{\partial_v} R^\rho{}_{\sigma\mu\nu} = 0$ always holds, identically, so \note{invariance along $\D u = \D v = 0$ is actually equivalent to invariance along the hypersurfaces $\D u = 0$.} The conditions \eqref{eq:invariance_of_Riemann_curvature} are equivalent to the statement that $\partial_a\partial_b\partial_c H = 0$ in Brinkmann coordinates \eqref{eq:reduced_pp_wave_metric2}, i.e. that $H(u,x)$ is a second order polynomial in $x^a$. In that case there always exists a coordinate transformation that removes the linear and constant terms (section 24.5 in 
\cite{stephani_kramer_maccallum_hoenselaers_herlt_2003}) so that the metric can be written as
\begin{align}\label{eq:plane_wave_brinkmann}
A = -2\D u \D v + A_{ab}(u)x^ax^b\, \D u^2+\delta_{ab}\, \D x^a \D x^b
\end{align}
This is the standard expression for a plane-wave metric in Brinkmann form. Moreover, an argument very similar to the one given in the previous subsection, shows that we may assume without loss of generality that the 1-form $\beta = \D u$ remains unchanged under this transformation. \\

Any such plane wave metric can also be written in Rosen form
\begin{align}\label{eq:Rosen_form}
\D s^2 = -2\D U\D V + h_{ij}(U)\D y^i \D y^j,
\end{align}
where $h_{ij}$ is a two-dimensional Riemannian metric. And conversely, any metric of Rosen form \eqref{eq:Rosen_form} can be cast in the form \eqref{eq:plane_wave_brinkmann}. The two coordinate systems are related via
\begin{align}
U = u,\quad V = v - \dfrac{1}{2}\dot E_{ai} E^i{}_b x^a x^b, \quad x^a = E^a{}_iy^i,
\end{align}
where $A_{ab} = \ddot E_{ai}E^i{}_b$ and $E^a{}_i$ is a vielbein for $h_{ij}$ in the sense that $h_{ij} = E^a{}_i E^b{}_j \delta_{ab}$, satisfying the additional symmetry condition $\dot E_{ai} E^i{}_b = \dot E_{bi} E^i{}_a$. Such a vielbein can always be chosen. For details we recommend the lecture notes \cite{Blau2011} by Matthias Blau and references therein (see also the Appendix of \cite{Blau_2003}). Note that we have momentarily labelled the $y$-coordinates by indices $i,j,k,\dots$ so as to distinguish them from indices $a,b,c,\dots$ in order that we may apply the usual notation with regards to the vielbein indices: $E^i{}_a$ represents the (matrix) inverse of $E^a{}_i$ and indices $a,b,c\dots$ are raised and lowered with $\delta_{ab}$, whereas indices $i,j,k,\dots$ are raised and lowered with $h_{ij}$. The dot that appears sometimes above the vielbein represents a $U$-derivative. Since the vielbein depends only on $U$, this derivative is equivalent to a $u$-derivative, and moreover the raising and lowering of the $a,b,c,\dots$ indices commutes with taking such a derivative of the vielbein.\\

It is again the case that, after relabeling $U,V\mapsto u,v$, the 1-form $\beta = \D u = \D U$ remains unchanged under this transformation, which in this case is easy to see. After also relabelling $y\mapsto x$, we conclude that we can express any \textit{Class 2} solution of plane-wave type in Rosen coordinates as follows,

\begin{align}
F = \alpha\, \phi(\beta/\alpha), \qquad A = -2\D u\D v + h_{ij}(u)\D x^i \D x^j, \qquad \beta = \D u \label{eq:ab_plane_wave},
\end{align}
where $\alpha = \sqrt{|A|}$. And conversely, for any choice of $\phi,h_{ij}(u)$, this is a vacuum solution to the field equations in Finsler gravity if $A$ is a vacuum solution to Einstein's field equation. The resulting geometry is of Berwald type with affine connection identical to the Levi-Civita connection of $\alpha$. 

\subsection{Linearized gravitational wave solutions}
\label{sec:linearized_Randers_sols}



The exact vacuum field equation for plane-wave metrics does not have a particularly nice expression in Rosen coordinates \eqref{eq:ab_plane_wave}. The \note{linearized} field equation, however, turns out to be very simple. So let's consider the scenario that the pseudo-Riemannian metric $\alpha$ is very close to the Minkowski metric. In this case we may write $h_{ij}(u) = \delta_{ij}+\varepsilon f_{ij}(u)$ with $\varepsilon \ll  1$. The \note{linearized} field equations (i.e. to first order in $\varepsilon$) for $\alpha$ then simply read\footnote{The full \note{linearized} vacuum field equation \eqref{eq:BEFEs} for $F$ is more complicated in general, but as discussed extensively above, if the vacuum field equation for $\alpha$ is satisfied then so is the field equation for $F$. In the case of Randers metrics, to which we will turn momentarily, the field equation for $F$ is even equivalent to the field equation for $\alpha$. Hence for our present purposes the field equations for $\alpha$ suffice.}
\begin{align}
f_{11}''(u) + f_{22}''(u)=0.
\end{align}
Hence $f_{11}$ and $-f_{22}$ must be equal up to an affine function of $u$. Here we will focus on the case where $f_{11}=-f_{22}$, which can always be achieved by means of the transverse traceless gauge\footnote{We leave open the question whether the form of the 1-form $\beta=\D  u$ always remains invariant under such a transformation to the transverse traceless gauge.}. Conventionally one writes the subscripts as $f_{11} = -f_{22} \eqqcolon f_+$ and $f_{12} \eqqcolon f_\times$, denoting the plus and cross polarization of the gravitational wave, so we will stick to that notation from here onwards. 
That brings us to the following expression that describes Finslerian gravitational waves of \ab-type:

\begin{align}\label{eq:ab_grav_waves}
F = \alpha\, \phi(\beta/\alpha), \qquad\left\{\begin{array}{ll}
 A = -2\D u \D v + (1+\varepsilon f_+(u)) \D x^2 + (1-\varepsilon f_+(u))\D y^2 + 2\varepsilon f_\times (u) \D x\,\D y \\ 
\beta = \D u
\end{array}\right.
\end{align} 

Note that if we substitute $u= (t-z)/\sqrt{2}$ and $v = (t+z)/\sqrt{2}$, then $A$ reduces to the standard expression for a gravitational wave metric in GR, i.e. 
\begin{align}
F = \alpha\, \phi(\beta/\alpha), \qquad\left\{\begin{array}{ll}
 A = -\D t^2 + (1+\varepsilon f_+(t-z)) \D x^2 + (1-\varepsilon f_+(t-z)\D y^2 + 2\varepsilon f_\times (t-z) \D x\,\D y+ \D z^2 \\ 
\beta = \frac{1}{\sqrt{2}}\left(\D t - \D z\right)
\end{array}\right. ,
\end{align}
for any choice of the function $\phi$.

%

\subsection{Linearized $(\alpha,\beta)$-metrics are Randers metrics}
\label{sec:lin_ab_metric_is_randers}

It is natural to linearlize not only in $\varepsilon$, characterizing the departure from flatness, but to also use a perturbative expansion in the `size' of the 1-form, characterizing the departure from GR and pseudo-Riemannian geometry. The physical intuition here is that, seeing how well GR works in most regimes, the most interesting class of Finsler spacetimes constists of those ones that are very close to GR spacetimes. The purpose of this section is to highlight that any \ab-metric is perturbatively equivalent to a Randers metric, to first order, so that from the physics point of view, Randers metrics are actually quite a bit more general than they might seem at first glance. After pointing this out we will turn our focus exclusively to Randers metrics for the remainder of the article.\\

So consider an \ab-metric constructed form a pseudo-Riemannian metric $\alpha$ and a 1-form $\beta$ such that $\beta\ll 1$. To see what happens in such a scenario, we replace $\beta$ with $\lambda\beta$ and expand to first order in $\lambda$. Then we obtain
\begin{align}
F = \alpha \phi\left(\frac{\lambda\beta}{\alpha}\right) \approx \alpha \left(\phi(0)+ \lambda \phi'(0)\frac{\beta}{\alpha}\right) = \alpha \phi(0) + \lambda\phi'(0)\beta = \tilde\alpha + \tilde\beta.
\end{align}
Hence to first order in $\lambda$, any $(\alpha,\beta)$-metric is indeed equivalent to a Randers metric\footnote{Actually this is not true for \textit{all} $(\alpha,\beta)$-metrics but only those which allow an expansion around $s = \beta/\alpha = 0$. This excludes Kropina metrics, for instance, because they are not well-behaved in the limit $\beta\to 0$.}. Consequently, by replacing $\D u$ by $\lambda\,\D u$ in \eqref{eq:ab_grav_waves}, which technically can be achieved by a coordinate transformation that scales $u$ by $\lambda$ and  $v$ by $1/\lambda$, it follows that to first order in $\lambda$ the Finsler metric of the \ab-type gravitational waves takes the form,

\begin{align}\label{eq:metric_naive_Randers_grav_wave0}
F = \alpha + \beta, \qquad\left\{\begin{array}{ll}
 A = -2\D u \D v + (1+\varepsilon f_+(t-z)) \D x^2 + (1-\varepsilon f_+(t-z)\D y^2 + 2\varepsilon f_\times (t-z) \D x\,\D y  \\ 
\beta =  \lambda\,\D u
\end{array}\right. .
\end{align}

The parameter $\lambda$ then characterizes the departure from GR and pseudo-Riemannian geometry. We will assume without loss of generality that $\lambda> 0$. Finally, replacing also $u$ and $v$ by $t$ and $z$, according to $u= (t-z)/\sqrt(2)$ and $v = (t+z)/\sqrt{2}$, we can write the metric in the following way, which we will take as the starting point for the calculation of the radar distance in Section \ref{sec:radar_distance}.

\begin{align}\label{eq:metric_naive_Randers_grav_wave}
F = \alpha + \beta, \qquad\left\{\begin{array}{ll}
 A = -\D t^2 + (1+\varepsilon f_+(t-z)) \D x^2 + (1-\varepsilon f_+(t-z)\D y^2 + 2\varepsilon f_\times (t-z) \D x\,\D y+ \D z^2 \\ 
\beta = \frac{\lambda}{\sqrt{2}}\left(\D t - \D z\right)
\end{array}\right. 
\end{align}


\section{Modified Randers Metrics}
\label{sec:Randers}

Motivated by the argument above we will now turn our focus to the simplest properly Finslerian $(\alpha,\beta)$-metric, the Randers metric, conventionally defined as $F = \alpha+\beta$. We will argue that in order to have a physically \note{natural} causal structure \note{(including, for instance, both a forward and a backward light cone, details follow below)}, the conventional definition must be modified slightly. It might seem to the reader that modifying the Randers metric would be in conflict with the spirit of the previous section, since to first order any \ab-metric should reduce to a Randers metric. It is important to note, however, that there is in principle the possibility that to different regions of the tangent bundle could correspond different Randers metrics. More precisely, we could define one \ab-metric $F_1$ on a conic subbundle $\mathcal A_1\subset TM\setminus 0$ and another \ab-metric, $F_2$, on a different conic subbundle $\mathcal A_2\subset TM\setminus 0$. If the two subbundles do not overlap then this defines a perfectly valid \ab-type Finsler spacetime on the union $\mathcal A = \mathcal A_1\cup\mathcal A_2$. To first order in the deviation from \mbox{(pseudo-)Riemannian} geometry this Finsler metric would reduce to a certain Randers metric on $\mathcal A_1$ and to a different Randers metric on $\mathcal A_2$. \note{On $\mathcal A$ as a whole, however, the resulting linearized metric might not be expressible as a single standard Randers metric. This is what we have in mind, and }our modification of the Randers metric, introduced below, is therefore completely consistent with the previous results. \note{While this paper was in review, a similar procedure was employed in \cite{heefer2023cosmological} to improve the causal structure of cosmological unicorn (i.e. non-Berwaldian Landsberg) solutions, based on the ideas proposed in this section.} \\

After a heuristic argument that motivates the desired modification, we show that our proposed version of the modified Randers metric has a very satisfactory causal structure. As a result of this a clear (future and past) timelike cone can be identified and within these timelike cones the signature of the Fundamental tensor is Lorentzian everywhere. The only constraint is that $b^2\equiv a^{\mu\nu}b_\mu b_\nu>-1$, which, interestingly, is in some sense the opposite of the condition $b^2<1$ that appears in the well-known positive definite case, see e.g. \cite{ChernShen_RiemannFinsler}. In one were to adopt the opposite signature convention to ours, however, the constraint in the Lorentzian case would also turn out to be $b^2<1$, matching the positive definite case.



\subsection{Motivation and definition}
\label{sec:Randers_modified1}


First of all, let us review why the definition of a Randers metric is not as clear in Lorentzian signature as it is in Euclidean signature. The original definition of a Randers metric, in positive definite Finsler geometry, is just $F = \alpha + \beta$, with $\alpha = \sqrt{a_{ij}y^i y^j}$ a Riemannian metric and $\beta = b_i y^i$ any 1-form\footnote{In order to satisfy all the axioms of a Finsler space, the 1-form must satisfy $|b|^2<1$, see e.g. \cite{ChernShen_RiemannFinsler}.}. This is well-defined as long as $\alpha$ is positive-definite, because in that case $A \equiv a_{ij}y^iy^j$ is always positive. If we allow $a_{ij}$ to be a Lorentzian metric, however, the quantity $A$ can become negative, in which case $\sqrt{A}$ is ill-defined, as we want $F$ to be a real function. One way to remedy this, at least at a technical level, is to restrict the conic subbundle $\mathcal A\subset TM\setminus 0$ to those vectors for which $a_{ij}y^iy^j>0$. 
This was the approach in e.g. \cite{Heefer_2021}, where it was shown that if $\mathcal A$ is defined as the forward timecone\footnote{We note that the signature convention in \cite{Heefer_2021} is the opposite as the one employed here, so in that case the condition $a_{ij}y^iy^j>0$ precisely select the timelike, not spacelike, vectors.} corresponding to $\alpha$, then under certain conditions on the 1-form $\beta$, such a Randers metric satisfies all axioms of a Finsler spacetime. The fact that $\mathcal A$ is restricted in this way, however, leads to issues when it comes to the physical interpretation. Here we take a different approach.\\ 

The obvious first alternative to restricting $\mathcal A$ to vectors with positive norm is to simply replace $A$ by $|A|$ and define $\alpha = \sqrt{|A|}$, as we have done throughout this article. In that case there's no need to restrict $\mathcal A$ to the timecone anymore. This leads to a Randers metric of the form $F = \sqrt{|A|} +\beta$. An undesirable consequence of this definition, however, is that light rays can only propagate into one half of the tangent space, namely the half given by $\beta<0$, which follows immediately from the null condition $F=0$ (see also  \cite{Voicu23_Finsler_ab_spacetime_condition}). In fact, the light cone separates the tangent space into only two connected components\footnote{This can be checked easily in suitable coordinates adapted to $\beta$.} and there is consequently not a straightforward interpretation in terms of timelike, spacelike and lightlike directions, at least not in the conventional way\footnote{We note that in the approach by Javaloyes and S\'anchez  \cite{Javaloyes2014-1, Javaloyes2014-2} a single, future pointing (by definition) cone is sufficient, though.}. 
We therefore take the viewpoint that outside of the half plane $\beta<0$ in each tangent space, this version of the Randers metric cannot be valid, and we need to modify it in that region. 
It is possible to remove the condition $\beta\leq 0$, extending the lightcone to the other half-plane $\beta>0$, by changing $F$ to $F=\sgn(A)\sqrt{|A|}+\sgn(\beta)\beta = \sgn(A)\sqrt{|A|}+|\beta|$. The result of this is that, under some mild assumptions (details will follow below) the single lightcone (from the $\beta<0$ half space) is mirrored to the complementary ($\beta\geq 0$) half space, whereas in the original half space intersected with the original cone of definition consisting of $\alpha$-timelike vectors, $F$ reduces to the standard Randers metric with an overall minus sign, $F = -(\alpha+\beta)$. This minus sign is not of any relevence, though, as the geometry is essentially determined by $F^2$. In particular, $F$ is now reversible, i.e. invariant under $y\to-y$. Notice also that we could have chosen a minus sign instead of a plus sign in the modified definition of $F$, but it turns out that in that case the resulting Finsler metric would not be guaranteed to have Lorentzian signature everywhere inside of the timelike cones\footnote{In case one employs the opposite signature convention $(+,-,-,-)$ the converse would be true. In that case the preferable choice would be $F=\sgn(A)\alpha-|\beta|$ rather than $F=\sgn(A)\alpha+|\beta|$.}. The present metric \textit{does} have this property as long as $b^2>-1$, and we discuss this in detail below.\\

\begin{defi}
Motivated by the preceding heuristic argument we define the \textit{modified Randers metric} as follows,
\begin{align}
F = \sgn(A)\alpha + |\beta|,\label{eq:modified_randers_metric}
\end{align}
where we recall for completeness that $\alpha = \sqrt{|A|}$,  $A = a_{ij}y^i y^j$ $\beta = b_i y^i$. 
\end{defi}
Both $\alpha$ and $A$ will sometimes be referred to as the (pseudo-)Riemannian metric, by a slight abuse of language, but it should always be clear from context what is meant.

\subsection{Causal structure}
\label{sec:Randers_causality}
Next we will show that the modified Randers metric \eqref{eq:modified_randers_metric} indeed has very nice properties. By definition, the light cone is given by 
\begin{align}
F = 0 \qquad \Leftrightarrow \qquad  \sgn(A)\leq 0 \,\&\, |A| = |\beta|^2 \qquad \Leftrightarrow \qquad  A =  -\beta^2.
\end{align}
It therefore follows that
\begin{align}
F=0 \qquad \Leftrightarrow\qquad (a_{\mu\nu}+ b_\mu b_\nu)\D x^\mu \D x^\nu = 0,
\end{align}
meaning that the light cone of $F$ is just the light cone of the auxilliary Lorentzian metric $\tilde a_{\mu\nu}(x) = a_{\mu\nu}+b_\mu b_\nu$. Indeed, the matrix determinant lemma guarantees that as long as $b^2 = a_{\mu\nu} b^\mu b^\nu>-1$ the metric $a_{\mu\nu}+ b_\mu b_\nu$ has Lorentzian signature, provided that $a_{\mu\nu}$ has Lorentzian signature. (For a proof see appendix \ref{sec:proof_of_signature}.)  This shows that as long as $b^2>-1$ the light cone separates the tangent space at each point into three connected components, which we can naturally interpret in the usual manner as the forward time cone, backward timecone, and the remainder consisting of spacelike vectors. Coincidentally we note that
\begin{align}
F<0 \qquad \Leftrightarrow\qquad (a_{\mu\nu}+ b_\mu b_\nu) y^\mu y^\nu < 0,
\end{align}
and hence it also follows that
\begin{align}
F>0 \qquad \Leftrightarrow\qquad (a_{\mu\nu}+ b_\mu b_\nu)y^\mu y^\nu > 0.
\end{align}
This leads to the additional convenience that $F$-timelike vectors are precisely given by $F<0$, and $F$-spacelike vectors by $F>0$, in addition to the null vectors being given, by definition, by $F=0$. We summarize these results in the following proposition.\\


\begin{prop}
    As long as $b^2>-1$, the causal structure of the modified Randers metric $F = \sgn(A)\alpha + |\beta|$ is identical to the causal structure of the \textit{Lorentzian} metric $a_{\mu\nu}+ b_\mu b_\nu$, with null vectors given by $F=0$, timelike vectors given by $F<0$, and spacelike vectors by $F>0$.
\end{prop}

As a result of these nice features of the causal structure of the modified Randers metric, it is possible to define time orientations in the usual manner, by means of a nowhere vanishing timelike vector field $T$. Such $T$ selects one of the two timelike cones as the `forward' one, namely the one that contains $T$. Then another timelike vector $y$ is future oriented (i.e. lies in the same forward cone as $T$) if and only if $(a_{\mu\nu}+ b_\mu b_\nu) T^\mu y^\nu  <0$. \note{We note that similar characterizations of time orientations in Finsler spacetimes (although not in terms of an auxilliary pseudo-Riemannian metric, like $a_{\mu\nu}+ b_\mu b_\nu$) have been discussed in \cite{Javaloyes2019}.} \\

In the special case that $\beta$ is covariantly constant with respect to $\alpha$ we have even more satisfactory results. In that case not only the causal structure but also the affine structure of $F$ can be understood in terms of $a_{\mu\nu}+b_\mu b_\nu$.

\begin{prop}
If $\beta$ is covariantly constant with respect to $\alpha$ and satisfies $b^2>-1$ then the causal structure and the affine structure of the modified Randers metric $F = \sgn(A)\alpha + |\beta|$ are identical to those of the Lorentzian metric $\tilde a_{\mu\nu} = a_{\mu\nu} + b_\mu b_\nu$. In other words, the timelike, spacelike and null geodesics of $F$ coincide with the timelike, spacelike and null geodesics of $\tilde a_{\mu\nu}$.
\end{prop}
\begin{proof}
The discussion above indicates that the causal structures coincide. It remains to show that also the affine structures concide in the case of a covariantly constant 1-form. This is again a result of the properties of $\tilde a_{\mu\nu}$. It can be shown (see Appendix  \ref{sec:proof_of_signature}) that the Christoffel symbols of $\tilde a_{\mu\nu}$ can be expressed in terms of the Christoffel symbols of $a_{\mu\nu}$ as 

\begin{align}
    \widetilde\Gamma^\rho_{\mu\nu} &= \Gamma^\rho_{\mu\nu} +\frac{1}{1+b^2}b^\rho \nabla_{(\mu}b_{\nu)}  - \left(a^{\rho\lambda} - \frac{1}{1+b^2}b^\rho b^\lambda\right)\left( b_\mu\nabla_{[\lambda} b_{\nu]} + b_\nu\nabla_{[\lambda} b_{\mu]}\right).
\end{align}
Hence it follows immediately that if $b_\mu$ is covariantly constant then $\widetilde\Gamma^\rho_{\mu\nu} = \Gamma^\rho_{\mu\nu}$ and the affine structure of $\tilde a_{\mu\nu}$ is the same as that of $a_{\mu\nu}$. We also known, by Prop. \ref{prop:coinciding_spray}, that the affine structure of $F$ is the same as that of $a_{\mu\nu}$. Hence the affine structure of $F$ is the same as that of $\tilde a_{\mu\nu}$.
\end{proof}

From this it follows immediately that the existence of radar neighborhoods is guaranteed \cite{Perlick2008}. \note{More precisely, given an observer and any spacetime event sufficiently close to the observer's worldline, there is (at least locally) exactly one future pointing light ray and one past pointing light ray that connect the event to the worldline of the observer.} This is of essential importance in our work, because what it essentially says is that the radar distance, calculated in Section \ref{sec:radar_distance}, is a well-defined notion. 

\subsection{Regularity and signature}

Given an \ab-metric of the form $F = \alpha \phi(s)$, with $s = \beta/\alpha$ and $\alpha = \sqrt{|A|}$, it can be shown that the determinant of the fundamental tensor is given by
\begin{align}\label{eq:ab_det_formula_main_text}
\det g_{ij} = \phi^{n+1}(\phi-s\phi')^{n-2}(\phi-s\phi' +  (\sgn(A) b^2-s^2)\phi'')\det a_{ij}.
\end{align}
The proof can be found in Appendix \ref{sec:ab_determinant}. Because of the appearance of $\sgn(A)$ the expression is slightly different from the well--known positive definite analogue, to which it reduces when $A>0$, i.e. $\sgn(A)=1$. For a modified Randers metric of the form $F = \text{sgn}(A)\alpha + |\beta|$ the function $\phi$ is given by $\phi(s) = \text{sgn}(A) + |s|$, so this reduces to 
\begin{align}
\frac{\det g}{\det a} =  \sgn(A)^{n-1}\left(\sgn(A)+|s|\right)^{n+1} = \left(\sgn(A)\frac{F}{\alpha}\right)^{n+1}.
\end{align}

Assuming the spacetime dimension $n$ is even, this means that $g$ has Lorentzian signature\footnote{The argument is the same as in the positive definitive case, using the same methods as those employed in Appendix \ref{sec:proof_of_signature}.} if and only if $\sgn(A)F>0$. Let us see what this entails. First note that $F<0$ trivially implies $A<0$. Hence $F<0$ implies Lorentzian signature. Before we move on, we should point out that this is a very satisfactory result. It means that within the entire timelike cone of $F$, the signature of the fundamental tensor is Lorentzian. Similarly, $A>0$ implies $F>0$. Hence $A>0$ also implies Lorentzian signature. What remains is the region where $A\leq 0$ and $F\geq 0$. Equivalently, $A\leq 0$ and $A+\beta^2\geq 0$. In this region, the determinant of the fundamental tensor is either undefined, is positive, or vanishes, so in any case the signature is not Lorentzian. But as this region lies outside the timelike cone, this is not a problem, as argued in section \ref{sec:Finsler_spacetimes_interpretation}.\\

It is helpful to think in terms of both the light cone of the metric $a_{ij}$ and the light cone of the metric $a_{ij}+b_ib_j$ (i.e. that of $F$). As mentioned previously, as long as $b^2>-1$, the latter metric is Lorentzian, provided the former is. That means its light cone is just a conventional one that we're familiar with from GR, just like the light cone of $a_{ij}$. The only region where the signature is \textit{not} Lorentzian, is precisely the region in between these two lightcones. Note that since $F<0$ implies $A<0$, the $F$-lightcone can never reach outside of the $a_{ij}$-light cone. The details depend on the causal character of the 1-form $\beta$ and are listed below. These properties can be checked easily by noting we may always choose coordinates such that at a given point $x\in M$ the metric $A$ has the form of the Minkowski metric and the 1-form $\beta$ has only one component (in the timelike or spacelike case) or two components (in the null case)\footnote{We recall that this can be seen as follows. First, since $a_{ij}$ is Lorentzian, it is always possible to choose coordinates such that $A$ is just the Minkowski metric at a given point $x\in M$. Writing $b^\mu = (b^0,b^1,\dots b^{n-1})$ in these coordinates, we may do a spatial rotation on the coordinates $b^1,\dots,b^{n-1}$, such that they are transformed into $(b^1,0,\dots,0)$, leaving the metric at $x$ unchanged. Then $b^\mu = (b^0,b^1,0,\dots,0)$. Now we separate the three cases. If $b^2=0$, it follows that $b^1=\pm b^0$ and by applying if necessary a spatial reflection in the $x^1$ direction we may choose either sign. If $b^2<0$ then we may go to the local rest frame by a Lorentz transformation, making $b^1=0$. If on the other hand $b^2>0$ we may perform a Lorentz transformation making $b^0=0$.}.

\begin{itemize}
\item If $\beta$ is null it is easily seen that the two lightcones intersect only for $y^\mu$ that are multiples of $b^\mu$. Thus their intersection spans a single line in the tangent space. 
\item If $\beta$ is timelike and $b^2>-1$ then the light cones do not intersect (apart from the trivial intersection in the origin).
\item If $\beta$ is spacelike (and assuming $\dim M>2$), then $a_{ij}$ induces a Lorentzian metric 
on the $(\dim M-1)$-dimensional hypersurface defined by $\beta=0$. In this case the two light cones intersect along the light cone of this induced Lorentzian metric. 
\item If $b^2 = -1$ there is only a single cone, namely the one corresponding to $\alpha$. The `light cone' corresponding to $F=0$ is now in fact a line, consisting all of multiples of $b^\mu$.  This case therefore does not have a viable physical interpretation.
\item If $b^2 < -1$ there is only a single cone, namely the one corresponding to $\alpha$. The `light cone' corresponding to $F=0$ is now non-existent, as $F=0$ has no solutions. This case therefore does not have a viable physical interpretation either.
\end{itemize}

\begin{figure}
	\centering
	\begin{subfigure}[b]{0.24\textwidth}
    	\includegraphics[width=\textwidth]{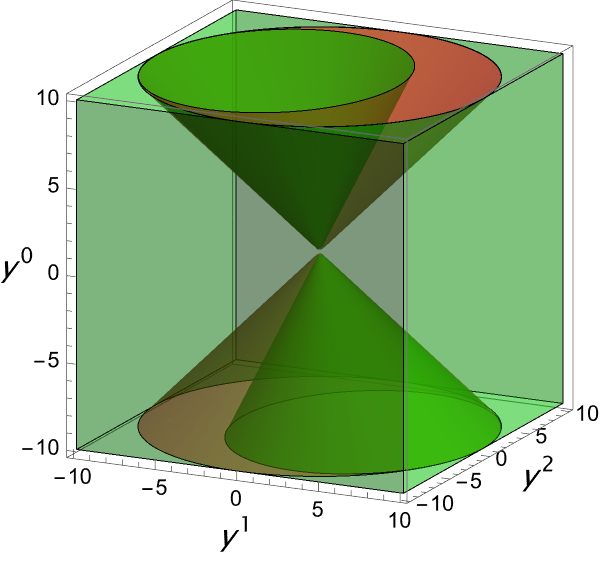}
    	\caption{Null 1-form with $\rho=0.6$}
    	\label{fig:L1}
    \end{subfigure}
    \hspace{20px}
    \begin{subfigure}[b]{0.24\textwidth}
    	\includegraphics[width=\textwidth]{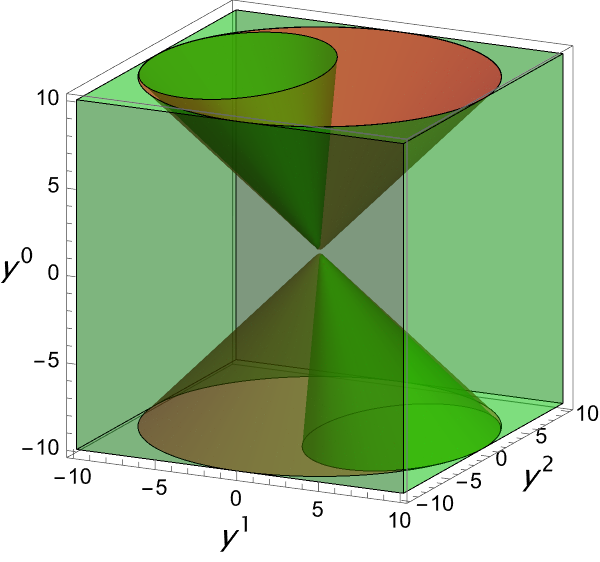}
    	\caption{Null 1-form with $\rho=1$}
    	\label{fig:L2}
    \end{subfigure}
    \hspace{20px}
    \begin{subfigure}[b]{0.24\textwidth}
    	\includegraphics[width=\textwidth]{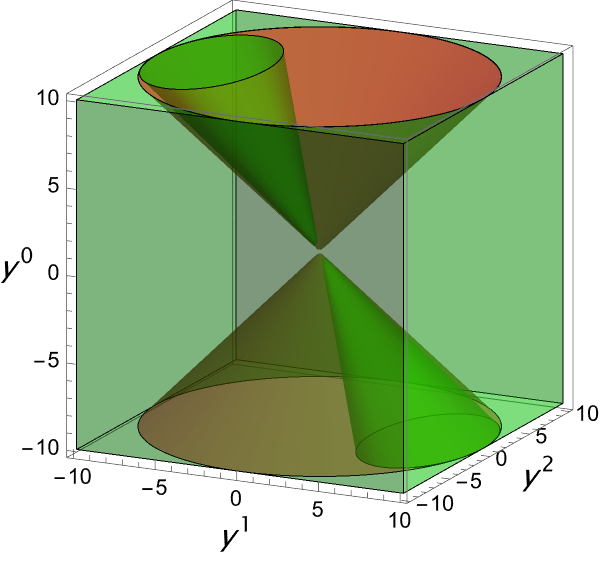}
    	\caption{Null 1-form with $\rho=1.4$}
    	\label{fig:L3}
    \end{subfigure}
    \begin{subfigure}[b]{0.24\textwidth}
    	\includegraphics[width=\textwidth]{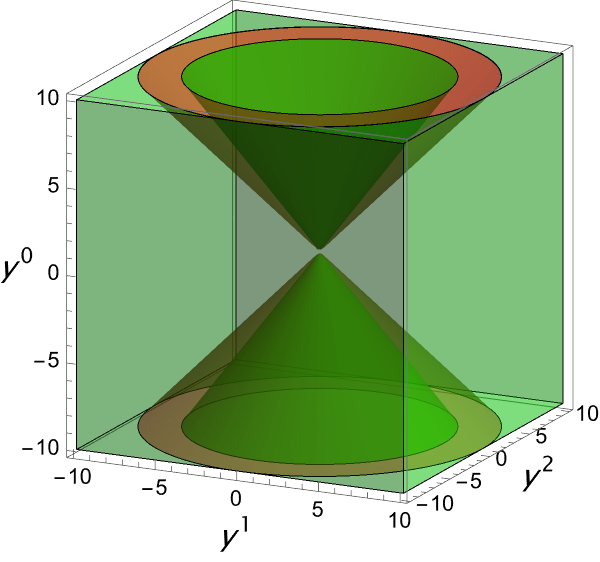}
    	\caption{Timelike 1-form with $\rho=0.65$}
    	\label{fig:T1}
    \end{subfigure}
    \hspace{20px}
    \begin{subfigure}[b]{0.24\textwidth}
    	\includegraphics[width=\textwidth]{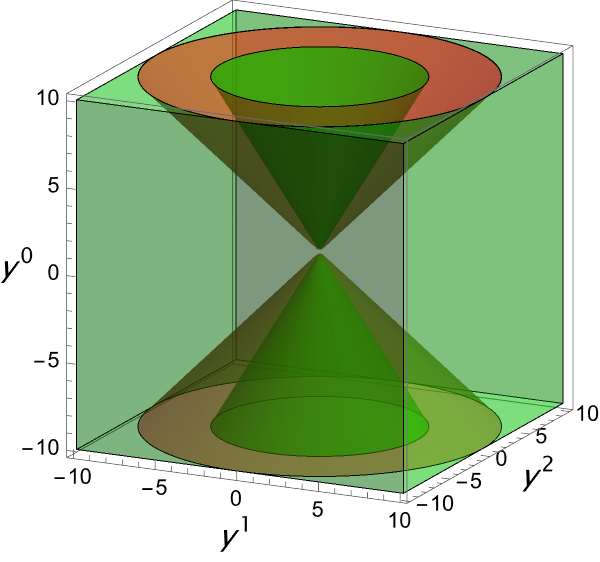}
    	\caption{Timelike 1-form with $\rho=0.8$}
    	\label{fig:T2}
    \end{subfigure}
    \hspace{20px}
    \begin{subfigure}[b]{0.24\textwidth}
    	\includegraphics[width=\textwidth]{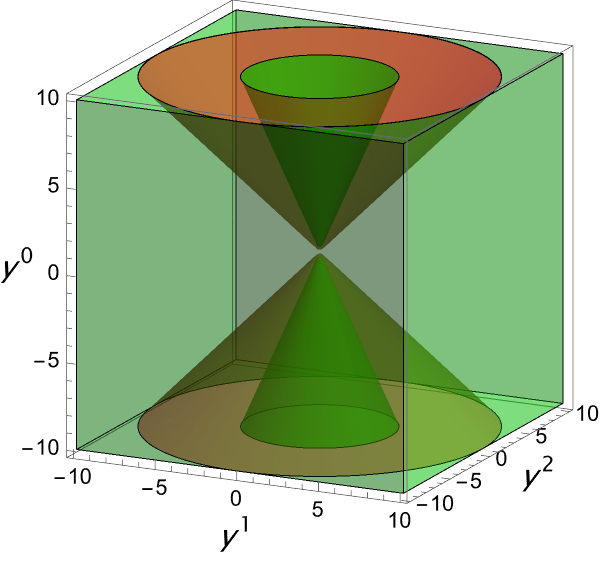}
    	\caption{Timelike 1-form with $\rho=0.9$}
    	\label{fig:T3}
    \end{subfigure}
    \begin{subfigure}[b]{0.24\textwidth}
    	\includegraphics[width=\textwidth]{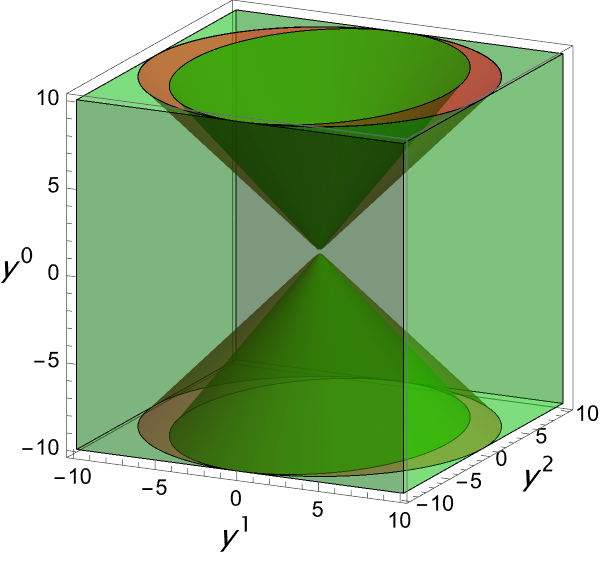}
    	\caption{Spacelike 1-form with $\rho=0.8$}
    	\label{fig:S1}
    \end{subfigure}
    \hspace{20px}
    \begin{subfigure}[b]{0.24\textwidth}
    	\includegraphics[width=\textwidth]{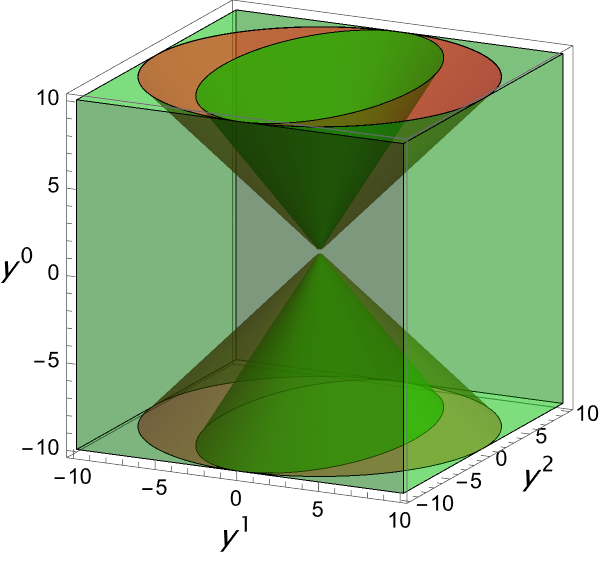}
    	\caption{Spacelike 1-form with $\rho=1.4$}
    	\label{fig:S2}
    \end{subfigure}
    \hspace{20px}
    \begin{subfigure}[b]{0.24\textwidth}
    	\includegraphics[width=\textwidth]{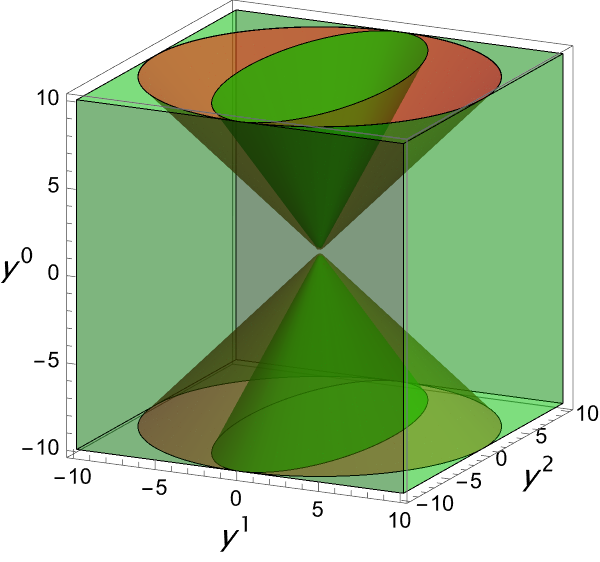}
    	\caption{Spacelike 1-form with  $\rho=2$}
    	\label{fig:S3}
    \end{subfigure}
    \begin{subfigure}[b]{0.24\textwidth}
    	\includegraphics[width=\textwidth]{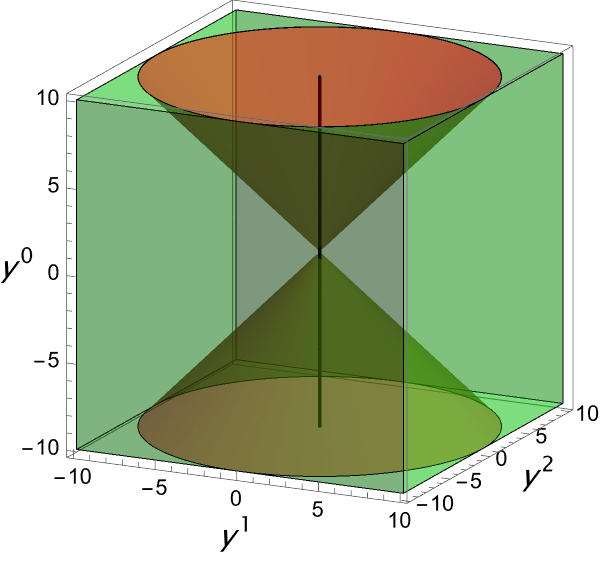}
    	\caption{Timelike 1-form with $b^2=-1$.}
    	\label{fig:T5}
    \end{subfigure}
    \hspace{20px}
    \begin{subfigure}[b]{0.24\textwidth}
    	\includegraphics[width=\textwidth]{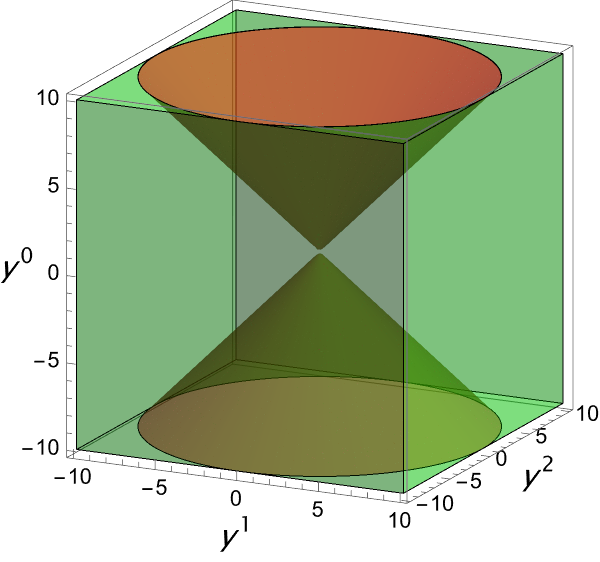}
    	\caption{Timelike 1-form with $b^2<-1$.}
    	\label{fig:T4}
    \end{subfigure}
    \caption{The figures show the lightcone and the signature of the fundamental tensor of $F = \sgn(A)\alpha + |\beta|$, where $A = -(y^0)^2+(y^1)^2+(y^2)^2+(y^3)^2$ and $\beta = \rho(y^0+y^1)$ (in the null case) or $\beta = \rho \,y^0$ (in the timelike case) or $\beta = \rho \,y^1$ (in the spacelike case) for several representative values of $\rho$, shown in the tangent space $T_xM$ at any point $x\in M$, at $y^3=0$. Green regions correspond to Lorentzian signature, red regions to non-Lorentzian signature. Fig. \ref{fig:L1} - \ref{fig:S3} show the physically reasonable scenarios, where $b^2>-1$. In that case two cones can be observed. The inner cone is the true light cone of $F$ (i.e. the set $F=0$), and the outer cone is the light cone of $a_{ij}$ (i.e. the set $A=0$). The only region with non-Lorentzian signature is precisely the gap in between the two cones. If on the other hand $b^2=-1$ (Fig. \ref{fig:T5}) then the light `cone' of $F$ is the line $y^1=y^2=y^3=0$. And if $b^2<-1$ (Fig. \ref{fig:T4}) then the light `cone' of $F$ consists only of the origin. Therefore we deem the latter two cases not physically interesting.}
\label{fig:lightcones_and_signatures}
\end{figure}

To get a better idea, Fig. \ref{fig:lightcones_and_signatures} displays the lightcones and the regions (in green) where the signature of the fundamental tensor is Lorentzian, for the modified Randers metric 
\begin{align}
F = \sgn(A)\alpha + |\beta| ,\qquad A = -(\D x^0)^2+(\D x^1)^2+(\D x^2)^2,\qquad \beta =\left\{\begin{array}{ll}
\rho \,\D x^0 & \text{if \textit{timelike}}\\ 
\rho\,(\D x^0+\D x^1) &\text{if \textit{null}}\\
\rho \,\D x^1 & \text{if \textit{spacelike}}
\end{array}\right. ,
\end{align}
for a number of representative values of the parameter $\rho$. In each subfigure, the inner lightcone is that of $F$ and the outer lightcone that of $A$. Note that for any $a_{ij}$ and $b_i$, it is always possible at any given point $x\in M$, to choose coordinates in such a way that $A$ and $\beta$ have the above form (or rather their analog in the relevant spacetime dimensionality). The following proposition summarizes these results.\\


\begin{prop}
    As long as $b^2>-1$, the signature of the fundamental tensor of $F = \sgn(A)\alpha + |\beta|$ is Lorentzian within the entire timelike cone, which is given by $F<0$. Immediately outside of the timelike cone there is a region that does not have Lorentzian signature, and further away (namely when $A>0$) the signature becomes Lorentzian again. When $b^2\leq -1$ the timelike cone is the empty set, so this case is not physically interesting.
\end{prop}

Since we only require Lorentzian signature within the timelike cone, these results are very satisfactory. \note{We point out in particular that this is true even when the 1-form is spacelike. On the other hand, the \textit{classical} Randers metric can only be considered a physically reasonable Finsler spacetime if the 1-form is either null or timelike  \cite{Voicu23_Finsler_ab_spacetime_condition}.} Finally, regarding the regularity of $F$, clearly $F$ is smooth everywhere except when $A=0$ or $\beta=0$. In particular, the set where $F$ is not smooth has measure zero. 

\section{Radar Distance for a Finsler Gravitational Wave}
\label{sec:radar_distance}

Now we are finally in the position to analyze the physical effects of a passing Finslerian gravitation wave of \ab-type. We have seen in Section \ref{sec:lin_ab_metric_is_randers} that, to first order, an \ab-metrics is equivalent to a Randers metric. And we have argued in section \ref{sec:Randers} that this should not be the standard Randers metric but rather our modified Randers metric. Thus our starting point will be the \note{linearized} gravitational wave solution of modified-Randers type. That is, we are interested in the solution  \eqref{eq:metric_naive_Randers_grav_wave} but with the conventional Randers metric $F=\alpha+\beta$ replaced by the modified Randers metric $F = \sgn(A)\alpha + |\beta|$, where $\alpha = \sqrt{|A|}$, $A = a_{ij}y^iy^j$. Note that this modification does not change any of the results pertaining to classification of solutions to the field equations, by the argument given at the beginning of section \ref{sec:Randers} that a modified Randers metric is `locally' equivalent to a standard Randers metric, in a certain precise sense. The relevant Finsler metric is therefore given by

\begin{align}\label{eq:metric_Randers_grav_wave}
F = \sgn(A)\alpha + |\beta|, \qquad\left\{\begin{array}{ll}
 A = -\D t^2 + (1+\varepsilon f_+(t-z)) \D x^2 + (1-\varepsilon f_+(t-z))\D y^2 + 2\varepsilon f_\times (t-z) \D x\,\D y+ \D z^2 \\ 
\beta = \frac{\lambda}{\sqrt{2}}\left(\D t - \D z\right)
\end{array}\right. 
\end{align}
%
%
Since actual gravitational wave measurements are done with interferometers, that effectively measure the \textit{radar distance}, the aim of this section is to compute that radar distance during the passing of a gravitational wave of the form \eqref{eq:metric_Randers_grav_wave}. \\

The setup is as follows. A light ray is emitted from some spacetime location with coordinates $(t_0, x_0, y_0, z_0)$, travels to another location in spacetime with coodinates $(t_0+\Delta t, x_0+\Delta x, y_0 + \Delta y, z_0+\Delta z)$, where it is reflected and after which it travels back to the original (spatial) location, with spacetime coordinates $(t_0+\Delta t_{\text{tot}}, x_0, y_0, z_0)$, being received there again. We are interested in the amount of proper time that passes between emission and reception of the light ray, as measured by an  `inertial' observer\footnote{In this context, we say that an observer is intertial if it would considered an intertial observer in the absence of the wave (i.e. when $f_+ = f_\times=0$). In other words, thinking of the gravitational wave as having a finite duration as it passes the Earth, an observer is inertial precisely if it is inertial before and after the wave passes.} located at spatial coordinates $(x_0, y_0, z_0)$. Because light travels forwards \textit{and} backwards during this time interval, one half of the time interval is usually called the \textit{radar distance} between the two spacetime points (sometimes the value is multiplied by the velocity of light, $c$, which we have set to 1, so that it has the dimensions of distance). In other words, the radar distance can be expressed as $R = \Delta \tau/2$.\\

\note{The expression for the radar distance of a standard gravitational wave in GR has been obtained in \cite{Rakhmanov_2009} and our calculation below follows essentially the same methods. At each step of the calculation we will clearly point out the differences with corresponding situation in GR, so that it is clear where each of the Finslerian effects (three separate effects can be identified) enters precisely. In addition to linearizing in $\varepsilon$, we also use a perturbative expansion in $\lambda$, as argued for at the end of Section \ref{sec:linearized_Randers_sols}. In fact, instead of working to first order in $\lambda$, we will work to second order in the Finslerian parameter, as certain important Finslerian effects only enter at second order, as we will see. We also neglect terms of combined order $\varepsilon\lambda^2$ and higher.}\\ 


\subsection{Finslerian null geodesics}\label{sec:geodesics}

\note{The first important observation here is that the geodesics in a Randers gravitational wave spacetime with Finsler metric $F = \text{sgn}(A)\alpha + |\beta|$ as in Eq. \eqref{eq:metric_Randers_grav_wave}  coincide with the geodesics of the GR spacetime with metric $\D s^2 = A$, because the affine connection of $F$ coincides with the Levi-Civita connection of $A$, by Prop.  \ref{prop:coinciding_spray}. For the derivation of the general form of geodesics, we may therefore assume the geometry is given by $\D s^2 = A$. Thus our point of departure is the metric 
\begin{align}
\D s^2 = A =-\D t^2 + (1+\varepsilon f_+(t-z)) \D x^2  (1-\varepsilon f_+(t-z))\D y^2 + 2\varepsilon f_\times (t-z) \D x\,\D y + \D z^2, \qquad \varepsilon\ll 1,
\end{align}
and we can essentially follow \cite{Rakhmanov_2009}. Using the coordinates $u= (t-z)/\sqrt{2}$ and $v = (t+z)/\sqrt{2}$ the geodesic equations to first order in $\varepsilon$ can be written as
\begin{align}
-\dot u \eqqcolon p_ v &= \text{const}, \label{eq:geod_1}\\
(1+\varepsilon f_+(u))\dot x +\varepsilon f_\times(u)\dot y\eqqcolon p_ x &= \text{const},\\
(1-\varepsilon f_+(u))\dot y +\varepsilon f_\times(u)\dot x\eqqcolon p_y &= \text{const}, \label{eq:geod_3}\\
\ddot v +\frac{1}{2} \varepsilon  \left(\dot{x}^2-\dot{y}^2\right) f_+'(u) + \varepsilon f_\times '(u)\dot x\dot y &= 0,
\end{align}
where \eqref{eq:geod_1}-\eqref{eq:geod_3} are obtained as first integrals, using the fact that the Lagrangian corresponding to $A$ is independent of the coordinates $v,x,y$. The first three equations can be rewritten to first order as
\begin{align}
\dot u = -p_v,\qquad \dot x = (1-\varepsilon f_+(u))p_x - \varepsilon f_\times(u)p_y,\qquad \dot y = (1+\varepsilon f(u))p_y- \varepsilon f_\times(u)p_x,
\end{align}
and can be integrated (with respect to an affine parameter $\sigma$, chosen without loss of generality such that $\dot u = 1$) to
\begin{align}
u = u_0 + \sigma,\quad x = x_0 + \sigma\left[\left(1 - \varepsilon \bar f_+(\sigma)\right)p_x- \varepsilon \bar f_\times(u)p_y\right],\quad y = y_0 + \sigma\left[\left(1 + \varepsilon \bar f_+(\sigma)\right)p_y- \varepsilon \bar f_\times(u)p_x\right],
\end{align}
where
\begin{align}
\bar f_{+,\times}(\sigma) \equiv \frac{1}{\sigma}\int_0^\sigma f_{+,\times}(u_0+\sigma)\D\sigma
\end{align}
is the averaged value of $f_{+,\times}$. The equation for $v$ can be integrated to
\begin{align}
\dot v = -\tilde p_{u0}-\frac{1}{2} \varepsilon  \left(p_x^2-p_y^2\right) f_+(u_0+\sigma) - \varepsilon f_\times (u_0+\sigma)p_x p_y
\end{align}
where $\tilde p_{u0} = p_{u0}-\frac{\varepsilon}{2}(p_x^2-p_y^2)f(u_0)- \varepsilon f_\times (u_0)p_x p_y$, $p_u = -\dot v$ (not necessarily constant) and $p_{u0}$ is its initial value at $\sigma = 0$. Integrating once again, we obtain
\begin{align}
v = v_0 - \tilde p_{u0}\sigma-\frac{1}{2}\varepsilon (p_x^2-p_y^2)\sigma \bar f_+(\sigma)- \varepsilon \bar \sigma f_\times (\sigma)p_x p_y.
\end{align}

Any geodesic emanating from a given point $x_0^\mu$ can thus be described by the following parameterized path, for certain values of $p_x, p_y$ and $\tilde p_{u0}$:

\begin{align}
u(\sigma) &= u_0 + \sigma, \label{eq:geod_prim_u}\\
x(\sigma) &= x_0 + \sigma\left[\left(1 - \varepsilon \bar f_+(\sigma)\right)p_x- \varepsilon \bar f_\times(\sigma)p_y\right],\label{eq:geod_prim_x}\\
y(\sigma) &= y_0 + \sigma\left[\left(1 + \varepsilon \bar f_+(\sigma)\right)p_y- \varepsilon \bar f_\times(\sigma)p_x\right],\label{eq:geod_prim_y}\\
v(\sigma) &= v_0- \tilde p_{u0}\sigma-\frac{1}{2}\varepsilon (p_x^2-p_y^2)\sigma \bar f_+(\sigma)- \varepsilon \sigma\bar f_\times (\sigma)p_x p_y.\label{eq:geod_prim_v}
\end{align}

We need to know the specific expression for \textit{null} geodesics, however. The modified null condition or modified dispersion relation (MDR) for massless particles, $F=0$, is the first place where the Finslerian character of the gravitational wave enters. According to Section \ref{sec:Randers_causality}, the condition $F=0$ is equivalent to $A = -\beta^2$, i.e.
\begin{align}
    -2\dot u\dot v + (1+\varepsilon f_+(u))\dot x^2 + (1-\varepsilon f_+(u))\dot y^2 +2\varepsilon f_\times(u) \dot x\dot y= -\beta^2 = -\lambda^2\dot u^2,
\end{align}
which, after substituting \eqref{eq:geod_prim_u}-\eqref{eq:geod_prim_v}, becomes $2\tilde p_{u0} + p_x^2+p_y^2=-\lambda^2$. We may therefore eliminate $\tilde p_{u0}$ and directly substitute this into the expression \eqref{eq:geod_prim_v} for $v(\sigma)$. A null geodesic starting at $(u_0,x_0,y_0,v_0)$ at $\sigma=0$ can therefore be described by the following parameterized path,
\begin{align}
u &= u_0+\sigma, \\
x &= x_0 + \sigma\left[\left(1 - \varepsilon \bar f_+(\sigma)\right)p_x- \varepsilon \bar f_\times(u)p_y\right],\\
y &= y_0 + \sigma\left[\left(1 + \varepsilon \bar f_+(\sigma)\right)p_y- \varepsilon \bar f_\times(u)p_x\right],\\
v &= v_0 + \frac{\sigma}{2}(p_x^2+p_y^2+\lambda^2)-\frac{1}{2}\varepsilon (p_x^2-p_y^2)\sigma \bar f_+(\sigma)- \varepsilon \sigma \bar f_\times (\sigma)p_x p_y.
\end{align}

Here we can make two important observations:
\begin{enumerate}
\item the effect due to the MDR or modified null condition enters at order $\lambda^2$;
\item in the limit $\lambda\to 0$ we recover the null geodesics for a standard gravitational wave in GR \cite{Rakhmanov_2009}.
\end{enumerate}
}

\subsection{Radar distance}

\note{Next we plug in the boundary conditions at the receiving point, $(u_0+\Delta u, x_0+\Delta x, y_0 + \Delta y, v_0+\Delta v)$. Note that $\sigma = \Delta u$ at that point, and hence from the middle two equations we infer that
\begin{align}
p_x = \frac{\Delta x}{\Delta u}\left(1 + \varepsilon \bar f_+(\Delta u)\right) +  \varepsilon \bar f_\times(\Delta u)\frac{\Delta y}{\Delta u}, \qquad p_y = \frac{\Delta y}{\Delta u}\left(1 - \varepsilon \bar f_+(\Delta u)\right)+  \varepsilon \bar f_\times(\Delta u)\frac{\Delta x}{\Delta u},
\end{align}
Plugging this into the $v$ equation yields
\begin{align}
2\Delta u\Delta v = \Delta x^2 (1+\varepsilon \bar f_+(\Delta u)) + \Delta y^2 (1-\varepsilon \bar f_+(\Delta u)) + 2\varepsilon \bar f_\times \Delta x\Delta y + \lambda^2\Delta u^2,
\end{align}
or equivalently,
\begin{align}
\left(1-\frac{\lambda^2}{2}\right)\Delta t^2 = \left(1 + \varepsilon \bar f(\Delta u)\right)\Delta x^2 + \left(1 - \varepsilon \bar f(\Delta u)\right)\Delta y^2 + 2\varepsilon \bar f_\times(\Delta u) \Delta x\Delta y +\left(1+\frac{\lambda^2}{2}\right)\Delta z^2 - \lambda^2\Delta z\Delta t,
\end{align}
where we have used that $-2\Delta u\Delta v = -\Delta t^2+\Delta z^2$.
This equation is solved to first order in $\varepsilon$ and $\lambda^2$ (neglecting $\varepsilon\lambda^2$ terms) by\footnote{In addition to this solution there is, formally, another solution to the equation. However, this other solution has the wrong zeroth order term, namely a negative one, which renders it physically irrelevant.}
\begin{align}\label{eq:delta_t_outgoing}
\Delta t = \Delta \ell + \left(\frac{\Delta x^2 - \Delta y^2 }{2\Delta\ell}\right)\varepsilon \bar f_+(\Delta u) + \left(\frac{\Delta x \Delta y }{\Delta\ell}\right)\varepsilon \bar f_\times(\Delta u)+ \frac{1}{2}\left(\frac{\Delta x^2 + \Delta y^2 + 2 \Delta z^2}{2\Delta\ell} - \Delta z\right)\lambda^2,
\end{align}
where $\Delta \ell \equiv \sqrt{\Delta x^2 + \Delta y^2 + \Delta z^2}$.\\

The right hand side in principle still depends on $t$ though, via $\bar f(\Delta u)$, so this is not yet a closed formula for $\Delta t$. However, since $\bar f$ only appears together with $\varepsilon$, and since we are only interested in the first order expression for $\Delta t$, any zeroth order expression for $\bar f$ suffices in this formula. We have

\begin{align}
\bar f(\Delta u)  &= \frac{1}{\Delta u}\int_0^{\Delta u} f(u_0+\sigma)\D\sigma = \frac{\sqrt{2}}{\Delta t - \Delta z}\int_0^{(\Delta t - \Delta z)/\sqrt{2}} f(u_0+\sigma)\D\sigma \\
&= \frac{\sqrt{2}}{\Delta \ell - \Delta z}\int_0^{(\Delta \ell - \Delta z)/\sqrt{2}} f(u_0+\sigma)\D\sigma + \mathcal O(\varepsilon)\\
&= \frac{\sqrt{2}}{\Delta \ell - \Delta z}\int_0^{(\Delta \ell - \Delta z)/\sqrt{2}} f\left(\frac{1}{\sqrt{2}}(t_0-z_0)+\sigma\right)\D\sigma + \mathcal O(\varepsilon)\label{eq:avg_perturbation_zeroth_order}
\end{align}
since $\Delta t = \Delta \ell + \mathcal O(\varepsilon)$. We introduce another symbol for this expression, namely
\begin{align}\label{eq:avg_perturbation_zeroth_order_tz}
\bar f(\Delta \ell, \Delta z, t_0-z_0)\equiv \frac{\sqrt{2}}{\Delta \ell - \Delta z}\int_0^{(\Delta \ell - \Delta z)/\sqrt{2}} f\left(\frac{1}{\sqrt{2}}(t_0-z_0)+\sigma\right)\D\sigma,
\end{align} 
where the explicit display of the arguments serves to remind us that $\bar f$ depends only on $\Delta \ell, \Delta z$ and the initial value of $t-z$. Since $\varepsilon \bar f(\Delta u)  = \varepsilon \bar f(\Delta \ell, \Delta z, t_0-z_0) + \mathcal O(\varepsilon^2)$, it follows that we can rewrite Eq. \eqref{eq:delta_t_outgoing}, to first order in $\varepsilon$ and $\lambda^2$, as
\begin{align}
\Delta t = \Delta \ell &+ \left(\frac{\Delta x^2 - \Delta y^2 }{2\Delta\ell}\right)\varepsilon \bar f_+(\Delta \ell, \Delta z, t_0-z_0)+ \left(\frac{\Delta x \Delta y }{\Delta\ell}\right)\varepsilon \bar f_\times(\Delta \ell, \Delta z, t_0-z_0) \\
&+ \frac{1}{2}\left(\frac{\Delta x^2 + \Delta y^2 + 2 \Delta z^2}{2\Delta\ell} - \Delta z\right)\lambda^2, \label{eq:Randers_time_elapsed_single_trip}
\end{align}
which is a closed expression for the elapsed coordinate time $\Delta t$ interval for a light ray traveling a certain spatial coordinate distance, in terms of the spatial coordinate separations and the initial value of $t-z$. \\

Now let's consider the complete trip, from $x^\mu_0$ to $x^\mu_0 + \Delta x^\mu$ and `back'. The total coordinate time elapsed during this trip is the sum of the forward trip and the backward trip time intervals. Schematically:
\begin{align}
\Delta t_\text{tot} &= \Delta t(\Delta x,\Delta y,\Delta z,t_0-z_0) + \Delta t(-\Delta x, -\Delta y, -\Delta z,t_0+\Delta t-(z_0+\Delta z)),
%
\end{align}
since the spatial interval on the backward trip is simply minus the forward spatial interval, and the `initial' value of $t-z$ for the backward trip is just the final value $t_0 -z_0 + \Delta t -\Delta z$ corresponding to the forward trip. Plugging in \eqref{eq:Randers_time_elapsed_single_trip} yields
\begin{align}\label{eq:Randers_time_elapsed_total_trip}
\Delta t_\text{tot} = 2\Delta \ell &+ \varepsilon\left(\frac{\Delta x^2 - \Delta y^2 }{2\Delta\ell}\right) \bar f_{+,\text{tot}} + \varepsilon \left(\frac{\Delta x \Delta y }{\Delta\ell}\right) \bar f_{\times,\text{tot}} \nonumber\\
&+ \frac{1}{2}\lambda^2\left(\frac{\Delta x^2 + \Delta y^2 + 2 \Delta z^2}{\Delta\ell} \right),
\end{align}
where $\bar f_{+,\text{tot}} = \bar f_{+,\text{forward}} + \bar f_{+,\text{backward}}$ and similarly for the $\times$-polarization, in terms of the forward and backward averaged amplitudes, respectively, given by
\begin{align}
\bar f_{+,\times,\text{forward}} &= \bar f_{+\times,}(\Delta \ell, \Delta z, t_0-z_0)  \\
&= \frac{\sqrt{2}}{\Delta \ell - \Delta z}\int_0^{(\Delta \ell - \Delta z)/\sqrt{2}} f_{+,\times}\left(\frac{1}{\sqrt{2}}(t_0-z_0)+\sigma\right)\D\sigma, \label{eq:bar_f_forward}\\
\bar f_{+,\times,\text{backward}} &= \bar f_{+,\times}(\Delta \ell, -\Delta z, t_0-z_0 + \Delta t - \Delta z) \\
&= \frac{\sqrt{2}}{\Delta \ell + \Delta z}\int_0^{(\Delta \ell + \Delta z)/\sqrt{2}} f_{+,\times}\left(\frac{1}{\sqrt{2}}(t_0+\Delta\ell-z_0 - \Delta z)+\sigma\right)\D\sigma, \label{eq:bar_f_backward}
\end{align}
where in the last expression we have replaced $\Delta t$ by $\Delta \ell$ in the argument of $f_{+,\times}$, because to zeroth order this makes no difference, and only the zeroth order expression for $\bar f_{+,\text{backward}}$ is relevant because $\bar f_{+,\text{backward}}$ always appears multiplied with $\varepsilon$ in the expressions we care about, like $\Delta t_\text{tot}$.\\


Equation \eqref{eq:Randers_time_elapsed_total_trip} gives the total coordinate time elapsed during the trip forward and back. The next step in the calculation of the radar distance $R = \Delta \tau/2$ is to convert the coordinate time interval into the proper time interval measured by the stationary observer local to the emission and reception of the light ray.  This is where a second Finslerian effect enters. For such a stationary observer we have $x=y=z=const$ and hence the 4-velocity is given by $(\dot t,0,0,0)$, where we will assume without loss of generality that $\dot t>0$. The proper time measured by an observer is given by the Finslerian length along its worldline $\Delta \tau =-\int F\, \D \sigma$. If we use $\sigma = \tau$ as our curve parameter, differentiating with respect to it shows that $F$ should be normalized as $F=-1$. This is the Finsler equivalent of the fact that in GR the worldline of a particle parameterized proper-time should always satisfy $g_{\mu\nu}\dot x^\mu \dot x^\nu=-1$ (or $+1$, depending on the signature convention). In the case of our observer the condition becomes
\begin{align}
F = \text{sgn}(A)\alpha + |\beta| = \text{sgn}(-\dot t^2)\sqrt{|\dot t^2|} + \frac{|\lambda\dot t|}{\sqrt{2}} = -|\dot t| + \frac{|\lambda\dot t|}{\sqrt{2}} = \left(-1+\frac{\lambda }{\sqrt{2}}\right)\dot t \stackrel{!}{=} -1.
\end{align}
It follows that
\begin{align}\label{eq:proper_time_vs_coord_time}
\Delta \tau = \left(1-\frac{\lambda }{\sqrt{2}}\right)\Delta t_\text{tot}
\end{align}
along the worldine of the stationary observer. Plugging in Eq. \eqref{eq:Randers_time_elapsed_total_trip} and \eqref{eq:proper_time_vs_coord_time} into $R = \Delta \tau/2$ we conclude that, to first order in $\varepsilon$ and second order in $\lambda$, the radar distance is given by
\begin{align}
%
%
R = \left(1 - \frac{\lambda }{\sqrt{2}}\right)\Delta \ell &+ \varepsilon\left(1 - \frac{\lambda }{\sqrt{2}}\right)\left(\frac{\Delta x^2 - \Delta y^2 }{4\Delta\ell}\right) \bar f_{+,\text{tot}} +\left(1 - \frac{\lambda }{\sqrt{2}}\right)\left(\frac{\Delta x \Delta y }{2\Delta\ell}\right)\varepsilon \bar f_{\times,\text{tot}}  + \frac{\lambda^2}{4}\left(\Delta\ell + \frac{\Delta z^2}{\Delta\ell} \right).\label{eq:Randers_Radar_Distance}
\end{align}
This expresses the radar distance as a function of the spatial coordinate distances and the initial value of $t-z$ (the latter enters the expression via $\bar f_{+,\times,\text{tot}}$). In the limit $\lambda\to 0$ we recover the expression for the radar distance in the case of a standard gravitational wave in GR \cite{Rakhmanov_2009}:

\begin{align}
R = \Delta \ell + \varepsilon\left(\frac{\Delta x^2 - \Delta y^2 }{4\Delta\ell}\right) \bar f_{+,\text{tot}} + \varepsilon \left(\frac{\Delta x \Delta y }{2\Delta\ell}\right) \bar f_{\times,\text{tot}} + \mathcal O (\varepsilon^2).\label{eq:GR_radar_distance}
\end{align}
}

Before we move on, let us summarize in what ways the Finslerian parameter $\lambda$ has entered our analysis so far:
\begin{enumerate}
\item The null trajectories are altered due to the fact the Finsler metric induces a modified null condition or MDR. As a result, it takes a \textit{larger} coordinate time interval for a light ray to travel a given spatial coordinate distance. This effect works in all spatial directions, even the direction parallel to the propagation direction of the light ray. This effect enters at order $\lambda^2$.
\item The ratio of proper time and coordinate time is altered with the result that \textit{less proper time is experienced per unit coordinate time} . This effect enters at order $\lambda$.
\end{enumerate}

There is, however, a third way in which the parameter enters. Namely in the relation between the coordinate distance and radar distance \textit{in the absence of the wave}. For a gravitational wave in GR these conveniently coincide; in the case of our Randers waves they don't. The formula for the radar distance derived above refers merely to coordinates. In order to make sense of the result, we would like to express the right hand side in terms of measurable quantities, like the radar distances in the various directions in the absence of the wave. Employing Eq. \eqref{eq:Randers_Radar_Distance} we write
\begin{align}
\Delta X = \left(1 - \frac{\lambda }{\sqrt{2}}\right)\Delta x + \frac{\lambda^2}{4}\Delta x, \\
\Delta Y = \left(1 - \frac{\lambda }{\sqrt{2}}\right)\Delta y + \frac{\lambda^2}{4}\Delta y , \\
\Delta Z = \left(1 - \frac{\lambda }{\sqrt{2}}\right)\Delta z + \frac{\lambda^2}{2}\Delta z ,
\end{align}
for the radar distance in the $x,y$ and $z$ direction \textit{in the absence of the wave}, and 
\begin{align}
R_0 = \left(1 - \frac{\lambda }{\sqrt{2}}\right)\Delta \ell + \frac{\lambda^2}{4}\left(\Delta\ell + \frac{\Delta z^2}{\Delta\ell} \right),
\end{align}
for the radar distance \eqref{eq:Randers_Radar_Distance} in the relevant direction \textit{in the absence of the wave}. Eliminating the coordinate distances in favour of the physical radar distances by virtue of the inverse transformations, valid to second order in $\lambda$,
\begin{align}
\Delta x &= \Delta X\left(1 + \frac{\lambda }{\sqrt{2}} + \frac{\lambda^2}{4}\right)\\
\Delta y &= \Delta Y\left(1 + \frac{\lambda }{\sqrt{2}} + \frac{\lambda^2}{4}\right)\\
\Delta z &= \Delta Z\left(1 + \frac{\lambda }{\sqrt{2}}\right) \\
\Delta \ell &= R_0\left(1 + \frac{\lambda }{\sqrt{2}} + \frac{3}{4}\lambda^2\right) - \frac{\Delta z^2}{4 R_0}\lambda^2 \\
&= R_0\left(1 + \frac{\lambda }{\sqrt{2}} + \frac{\lambda^2}{4}\right) - \frac{\Delta Z^2}{4 R_0}\lambda^2
\end{align}
we can express the radar distance in the presence of the wave as

\begin{align}
\boxed{
R =  R_0 + \varepsilon\left(\frac{\Delta X^2 - \Delta Y^2 }{4R_0}\right)\bar f_{+,\text{tot}} + \varepsilon\left(\frac{\Delta X\Delta Y }{2R_0}\right)\bar f_{\times,\text{tot}}  + \mathcal O(\varepsilon^2, \lambda^3, \varepsilon\lambda^2).
}
\end{align}

This is a remarkable result. By expressing the radar distance in terms of the physical observables $\Delta X,\Delta Y$ and $R_0$ rather than merely coordinates, 
all dependence on $\lambda$ has disappeared to the desired order and the expression is identical to its GR counterpart, Eq. \eqref{eq:GR_radar_distance}! We must conclude, therefore, that the effect of a Randers gravitational wave on interferometer experiments is virtually indistinguishable from that of a conventional GR gravitational wave. \\

It is important to remark that by no means this implies that all phenomena in such a Finsler spacetime are identical to their GR counterparts. It might be possible to detect the presence of a non-vanishing $\lambda$ by some other means. This is a very interesting and important questions, however it is beyond the scope of this article and something to explore in future work. Our results pertain merely to gravitational wave effects as observed by interferometers.

\section{Discussion}
\label{sec:discussion}

The main aim of this paper was to study the physical effect of Finslerian gravitational waves and, in particular, to investigate the question if and how such waves can be distinguished, observationally, from the classical gravitational waves of general relativity. To this effect we have derived an expression for the radar distance at the moment a Finsler gravitational passes, say, the earth. This radar distance is the main observable that is measured by interferometers. Remarkably, we have found that the expression for the radar distance is indistinguishable from its non-Finslerian counterpart, leading us to conclude that interferometer experiments would not be able to distinguish between a general relativistic and a Finslerian gravitational wave, at least not with regards to the radar distance. This is on the one hand disappointing, since indicates means we cannot use such measurements to test the Finslerian character of our spacetime. On the other hand, though, it means that the current gravitational wave measurements are all compatible with the idea that spacetime has a Finslerian nature. To the best of our knowledge this is the first time an explicit expression for the Finslerian Radar length has been obtained for the case of finite spacetime separations, and as such our work may be seen as a proof of concept. Repeating the analysis for other Finsler spacetime geometries may lead to additional insight as to the observational signature of Finsler gravity. \\

\note{It is important to point out that Finslerian effects may also play a role in the \textit{generation} of gravitational waves in, say, a black hole merger event. This could lead to a Finslerian correction to the waveform and this could be measured in interferometer experiments, at least in principle. In order to be able to investigate this, however, Finslerian black hole solutions need to be better understood. A start in this direction has been made in \cite{cheraghchi2022fourdimensional}, where all 4-dimensional spherically symmetric Finsler metrics of Berwald type have been classified.}\\

The other parts of the article, leading up to the calculation of the radar length, were more mathematical in nature. We have introduced a class of exact solutions to the field equation in Finsler gravity that have a close resemblance to the well-known general relativistic pp-waves, and that generalize all of the pp-wave-type solutions currently known in the literature \cite{Fuster:2015tua, Fuster:2018djw, Heefer_2021}. These solutions are \ab-metrics, where $\alpha$ is a classical pp-wave and $\beta$ is its  defining covariantly constant null 1-form. Consequently our solutions are of Berwald type. Their linearized versions, we have shown, may be interpreted as Finslerian gravitational waves of modified Randers type.\\

Indeed, along the way we have introduced a small modification to the standard definition the Randers metric, motivated by the observation that the physical interpretation of the causal structure of the standard Randers metric is not immediately obvious. In contrast, we have shown that our modified Randers metrics have the nice property that their causal structure is completely equivalent to the causal structure of some auxiliary \mbox{(pseudo-)Riemannian} metric, hence leading to a perfectly clear physical interpretation. We stress that this auxilliary metric is different from the \textit{defining} \mbox{(pseudo-)Riemannian} metric $\alpha$. In the special case that the defining 1-form of the Randers metric is covariantly constant (which is the case, for example, for our solutions) we have even more satifactory results. In this case not only the causal structure, but also the affine structure of the Randers metric coincides with that of the auxilliary (pseudo)-Riemannian metric, i.e. the timelike, spacelike and null geodesics of the Finsler metric can be understood, respectively, as the timelike, spacelike and null geodesics of the auxiliary (pseudo)-Riemannian metric. A particularly nice consequence of this is the guaranteed existence of radar neighborhoods, i.e. that given an observer and any event in spacetime, there is (at least locally) exactly one future pointing light ray and one past pointing light ray that connect the worldline of the observer to the event. This is of essential importance in our work, because without this property it would have not been possible to perform the calculation of the radar distance in the last part of the article, simply because the notion of radar distance would not even make sense in that case. \\

Let us now point out some of the limitations of our investigation. First of all, it is by no means expected that the Finslerian gravitational waves discussed here should be only possible ones. Although being much larger than even the complete class of \textit{all} Lorentzian (i.e. non-Finslerian) geometries, the class of \ab-metrics of Berwald type, to which we have restricted our analysis, is still quite restrictive in the large scheme of (Finsler geometric) things. \note{Moreover, even within the class of $(\alpha,\beta)$-metrics, our analysis is only valid for those metrics that can be regarded as `close' to a Lorentzian metric, such that they can be approximated by Randers metrics.} So even though our results suggest that there is no observable difference between the Finslerian gravitational waves discussed in this article and their GR counterparts, there might be more general types of Finslerian gravitational waves that \textit{could} be distinguished observationally from the general relativistic ones by means of interferometer experiments. Furthermore, radar distance experiments are by no means the only way of probing our spacetime geometry. It might be possible to detect the Finslerian character of spacetime in some other way. We have not explored this possibility here, but we plan to investigate this in the future.\\

Moreover, we have assumed in our calculations that the amplitude of the gravitational waves as well as the Finslerian deviation from general relativity are sufficienty small such that a perturbative approach to first order in the former and second order in the latter is valid. It would be of interest to repeat the calculation to higher order in perturbation theory. We expect that this would in principle be a straightforward, yet possibly tedious, exercise.



\begin{acknowledgments}
S.H. wants to thank Rick Sengers and Nicky van den Berg for fruitful discussions and for their input with regards to the figures. S.H. also wants to thank Luc Florack for fruitful discussions, in particular his suggestions with regards to perturbation theory. We would like to acknowledge networking support by the COST Action CA18108, supported by COST (European Cooperation in Science and Technology).
\end{acknowledgments}
\appendix

\section{Some Properties of the Metric $a_{\mu\nu}+b_\mu b_\nu$}
\label{sec:proof_of_signature}

\subsection{Proof of Lorentzian signature}
Here we prove that if $a_{\mu\nu}$ has Lorentzian signature and $a_{\mu\nu}b^\mu b^\nu>-1$ then $\tilde a_{\mu\nu} = a_{\mu\nu}+b_\mu b_\nu$ also has Lorentzian signature. We write $b^2 = a_{\mu\nu}b^\mu b^\nu$. First, the matrix determinant lemma says that 
\begin{align}
\det \tilde a= (1+b^2)\det a.
\end{align}
As long as $b^2>-1$ this implies that $\det \tilde a$ has the same sign as $\det a$, so assuming $a_{\mu\nu}$ is Lorentzian, $\tilde a$ has negative determinant. In 4D this immediately implies that $\tilde a$ is Lorentzian (although the signs of the eigenvalues might be flipped with respect to $a_{\mu\nu}$). However, let's assume the dimensionality is arbitrary. Consider the family of 1-forms $b^{(\eta)}_\mu= \eta b_\mu$, where $\eta\in[0,1]$. For any $\eta$ we have
\begin{align}
\det \widetilde {a^{(\eta)}} = \left[1+\left(b^{(\eta)}\right)^2\right]\det a = \left[1+\eta^2 b^2\right]\det a,
\end{align}
$\det \widetilde {a^{(\eta)}}$ has the same sign for all values of $\eta$. 
Now since each of the $n$ eigenvalues of $\widetilde {a^{(\eta)}}$ can be expressed as continuous function of $\eta$, it follows that the respective signs of the $n$ eigenvalues cannot change when we change $\eta$. To see why, suppose that the $k$-th eigenvalue is positive for some $\eta_1$ and negative for some $\eta_2$. By the intermediate value theorem, there must exist some $\eta$ between $\eta_1$ and $\eta_2$ for which the eigenvalue vanishes. In that case the determinant vanishes for that value of $\eta$, which is a contradiction, as the determinant never vanishes as we have just seen. This argument proves that $\widetilde {a^{(\eta)}}$ has the same signature for all values of $\eta$, because the signs of the eigenvalues remain unchanged. In particular, $\tilde a = \widetilde {a^{(1)}}$ has the same signature as $a = \widetilde {a^{(0)}}$. Therefore, if $a_{\mu\nu}$ is Lorentzian and $b^2>-1$ then $\tilde a$ is Lorentzian as well.\\

\subsection{Affine structure}

Here we derive an explicit formula for the Christoffel symbols of the metric $\tilde a_{\mu\nu} = a_{\mu\nu}+b_\mu b_\nu$, where it is again assumed that $b^2>-1$.

\begin{prop}
The Christoffel symbols of $\tilde a_{\mu\nu}$ can be expressed as
\begin{align}
    \widetilde\Gamma^\rho_{\mu\nu} &= \Gamma^\rho_{\mu\nu} +\frac{1}{1+b^2}b^\rho \nabla_{(\mu}b_{\nu)}  - \left(a^{\rho\lambda} - \frac{1}{1+b^2}b^\rho b^\lambda\right)\left( b_\mu\nabla_{[\lambda} b_{\nu]} + b_\nu\nabla_{[\lambda} b_{\mu]}\right).
\end{align}
where $\nabla$ is the covariant derivative corresponding to $a_{\mu\nu}$
\end{prop}

We prove this below, but first we point out the following immediate consequence.

\begin{cor}
If $b_\mu$ is covariantly constant with respect to $a_{\mu\nu}$, the affine structure of $\tilde a_{\mu\nu}$ is the same as the affine structure of $a_{\mu\nu}$, i.e. $\widetilde\Gamma^\rho_{\mu\nu} = \Gamma^\rho_{\mu\nu}$.
\end{cor}

\begin{proof}
As long as $b^2>-1$ the formula for the determinant displayed above shows that $\tilde a_{\mu\nu} = a_{\mu\nu}+b_\mu b_\nu$ is invertible as a matrix. It can be easily checked that its inverse is given by 
\begin{align}
    \tilde a^{\mu\nu} =a^{\mu\nu} - \frac{1}{1+b^2}b^\mu b^\nu.
\end{align}
Unless otherwise specified (as in the case of $\widetilde\Gamma$ below!) indices are raises and lowered with $a_{\mu\nu}$. Denoting $\Gamma_{\lambda\mu\nu}=a_{\lambda\rho}\Gamma^\rho_{\mu\nu}$ and $\widetilde\Gamma_{\lambda\mu\nu}=\tilde a_{\lambda\rho}\widetilde\Gamma^\rho_{\mu\nu}$ we first note that we can express the latter as

\begin{align}
    \widetilde\Gamma_{\lambda\mu\nu} =\frac{1}{2} \left(
    \partial_\mu \tilde a_{\lambda\nu} + \partial_\nu \tilde a_{\mu\lambda} - \partial_\lambda \tilde a_{\mu\nu}\right) = 
\Gamma_{\lambda\mu\nu}  + b_\lambda\partial_{(\mu} b_{\nu)} - b_\mu\partial_{[\lambda} b_{\nu]} - b_\nu\partial_{[\lambda} b_{\mu]},
\end{align}
where $(\mu,\nu)$ denotes symmetrization and $[\mu,\nu]$ denotes anti-symmetrization. Therefore it follows that
\begin{align}
\widetilde\Gamma^\rho_{\mu\nu} &= \tilde a^{\rho\lambda}\widetilde\Gamma_{\lambda\mu\nu} = \left(a^{\rho\lambda} - \frac{1}{1+b^2}b^\rho b^\lambda\right)\left(\Gamma_{\lambda\mu\nu}  + b_\lambda\partial_{(\mu} b_{\nu)} -b_\mu\partial_{[\lambda} b_{\nu]} -b_\nu\partial_{[\lambda} b_{\mu]}\right) \\
&=\Gamma^\rho_{\mu\nu} - \frac{1}{1+b^2}b^\rho b_\lambda\Gamma^\lambda_{\mu\nu} 
  +\left(a^{\rho\lambda} - \frac{1}{1+b^2}b^\rho b^\lambda\right)b_\lambda\partial_{(\mu} b_{\nu)} - \left(a^{\rho\lambda} - \frac{1}{1+b^2}b^\rho b^\lambda\right)\left( b_\mu\partial_{[\lambda} b_{\nu]} + b_\nu\partial_{[\lambda} b_{\mu]}\right).
\end{align}
The second and third term add up to
\begin{align}
-\frac{1}{1+b^2}b^\rho & b_\lambda\Gamma^\lambda_{\mu\nu} +\left(a^{\rho\lambda} - \frac{1}{1+b^2}b^\rho b^\lambda\right)b_\lambda\partial_{(\mu} b_{\nu)} \\
&= -\frac{1}{1+b^2}b^\rho b_\lambda\Gamma^\lambda_{\mu\nu} +b^\rho\partial_{(\mu} b_{\nu)} - \frac{b^2}{1+b^2} b^\rho \partial_{(\mu} b_{\nu)} =\\
&=-\frac{1}{1+b^2}b^\rho b_\lambda\Gamma^\lambda_{\mu\nu} \ + \frac{1}{1+b^2} b^\rho \partial_{(\mu} b_{\nu)}\\
&=\frac{1}{1+b^2}b^\rho \left(\partial_{(\mu}b_{\nu)}-b_\lambda\Gamma^\lambda_{\mu\nu} 
\right)\\
&=\frac{1}{1+b^2}b^\rho \nabla_{(\mu} b_{\nu)}.
\end{align}
This shows that
\begin{align}
    \widetilde\Gamma^\rho_{\mu\nu} &= \Gamma^\rho_{\mu\nu} +\frac{1}{1+b^2}b^\rho \nabla_{(\mu}b_{\nu)}  - \left(a^{\rho\lambda} - \frac{1}{1+b^2}b^\rho b^\lambda\right)\left( b_\mu\partial_{[\lambda} b_{\nu]} + b_\nu\partial_{[\lambda} b_{\mu]}\right).
\end{align}
Finally, we may replace all partial derivatives with covariant ones because
\begin{align}
    \nabla_{[\lambda} b_{\nu]}= \nabla_\lambda b_\nu - \nabla_\nu b_\lambda = \partial_\lambda b_\nu - \Gamma^\mu_{\lambda\nu} b_\mu- \partial_\nu b_\lambda+\Gamma^\mu_{\nu\lambda} b_\mu = \partial_\lambda b_\nu- \partial_\nu b_\lambda  = \partial_{[\lambda} b_{\nu]}.
\end{align}
That yields the desired formula.
\end{proof}

\section{Determinant of a Not Necessarily Positive Definite $(\alpha,\beta)$-Metric}
\label{sec:ab_determinant}

Here we derive the formula Eq. \eqref{eq:ab_det_formula_main_text} for the determinant of a not necessarily positive definite \ab-metric, generalizing the well-known result from the positive definite case. More precisely, we consider Finsler metrics the form $F=\alpha\phi(s)$, where $s = \beta/\alpha$, $\alpha = \sqrt{|A|}=\sqrt{|a_{ij}y^i y^j|}$, $A = a_{ij}y^i y^j =  \text{sgn}(A)\alpha^2$, and where $a_{ij}$ is assumed to be a \mbox{(pseudo-)Riemannian} metric, i.e. not necessarily Riemannian/positive definite.\\

In complete analogy with the positive definite case, it can be shown by direct calculation that the fundamental tensor $g_{ij}\equiv \tfrac{1}{2}\bar\partial_i\bar\partial_j F^2$ is given by
\begin{align}\label{eq:alpha_beta_metric_fundamental_tensor}
g_{ij} = \sgn(A)\rho a_{ij} + \rho_0 b_i b_j + \rho_1(b_i\alpha_j + \alpha_i b_j) + \rho_2\alpha_i\alpha_j,
\end{align}
where we have defined $\alpha_i= a_{ij}y^j/\alpha$, and with coefficients given by
\begin{align}
\rho &= \phi(\phi-s\phi'),\\
\rho_0&= \phi \phi'' + \phi'\phi',\\
\rho_1 &=  -(s\rho_0 - \phi \phi') = -\left[s(\phi \phi'' + \phi'\phi') - \phi\phi '\right],\\
\rho_2 &=  -s\rho_1 = s\left[s(\phi \phi'' + \phi'\phi') - \phi\phi '\right].
\end{align}
The only difference here with the positive definite case is the factor sign$(A)$ appearing in the first term in Eq. \eqref{eq:alpha_beta_metric_fundamental_tensor}. Denoting $\dim M = n$ can write this in matrix notation as
\begin{align}
g = \sgn(A)\rho \left(a + UWV^T\right),
\end{align}
in terms of the three matrices
\begin{align}
W = \frac{\sgn(A)}{\rho}\mathbb I_{4\times 4},\qquad U = (\vec b, \vec b, \vec \alpha, \vec \alpha), \qquad V = (\rho_0 \vec b, \rho_1 \vec\alpha, \rho_1 \vec b, \rho_2 \vec\alpha).
\end{align}
$U$ and $V$ are both $n\times 4$ matrices. It is a well-known result (one of the matrix determinant lemmas, \note{see e.g. \cite{harville2008matrix}}) that assuming $a$ is an invertible matrix the determinant of the expression in brackets is equal to
\begin{align}
\det \left(a + UWV^T\right) = \det \left(\mathbb I_{4\times 4} + WV^Ta^{-1}U\right)\det a.
\end{align}
It follows that
\begin{align}
\det g = \sgn(A)^n\rho^n\det \left(\mathbb I_{4\times 4} + WV^Ta^{-1}U\right)\det a.
\end{align}
The matrix product $WV^Ta^{-1}U = \tfrac{\sgn(A)}{\rho}V^Ta^{-1}U$ can be evaluated by explicit computation and reads
\begin{align}
WV^Ta^{-1}U = \frac{\sgn(A)}{\rho}\left(
\begin{array}{cccc}
 b^2 \text{$\rho_0 $} & b^2 \text{$\rho_0$} & \sgn(A) s \text{$\rho_0$} & \sgn(A) s \text{$\rho_0$} \\
 \sgn(A) s \rho  & \sgn(A) s \rho  & \sgn(A) \rho  & \sgn(A) \rho  \\
 b^2 \text{$\rho_1$} & b^2 \text{$\rho_1$} & \sgn(A) s \text{$\rho_1$} & \sgn(A) s \text{$\rho_1$} \\
 \sgn(A) s \text{$\rho_2$} & \sgn(A) s \text{$\rho_2$} & \sgn(A) \text{$\rho_2$} & \sgn(A) \text{$\rho_2$} \\
\end{array}
\right).
\end{align}
Hence we obtain
\begin{align}
\det g&= \sgn(A)^n\rho^n \det\left(\mathbb I_{4\times 4}+\frac{\sgn(A)}{\rho}\left(
\begin{array}{cccc}
 b^2 \text{$\rho_0 $} & b^2 \text{$\rho_0$} & \sgn(A) s \text{$\rho_0$} & \sgn(A) s \text{$\rho_0$} \\
 \sgn(A) s \rho  & \sgn(A) s \rho  & \sgn(A) \rho  & \sgn(A) \rho  \\
 b^2 \text{$\rho_1$} & b^2 \text{$\rho_1$} & \sgn(A) s \text{$\rho_1$} & \sgn(A) s \text{$\rho_1$} \\
 \sgn(A) s \text{$\rho_2$} & \sgn(A) s \text{$\rho_2$} & \sgn(A) \text{$\rho_2$} & \sgn(A) \text{$\rho_2$} \\
\end{array}
\right)\right)\det a \\
&= \sgn(A)^n\rho^n \det\left(\mathbb I_{4\times 4}+\frac{1}{\rho}\left(
\begin{array}{cccc}
 \sgn(A) b^2 \text{$\rho_0 $} & \sgn(A) b^2 \text{$\rho_0$} &  s \text{$\rho_0$} & s \text{$\rho_0$} \\
  s \rho  &  s \rho  &  \rho  & \rho  \\
 \sgn(A) b^2 \text{$\rho_1$} & \sgn(A) b^2 \text{$\rho_1$} & s \text{$\rho_1$} & s \text{$\rho_1$} \\
s \text{$\rho_2$} & s \text{$\rho_2$} &  \text{$\rho_2$} &  \text{$\rho_2$} \\
\end{array}
\right)\right)\det a \\
&= \phi^{n+1}(\phi-s\phi')^{n-2}(\phi-s\phi' +  (\sgn(A) b^2-s^2)\phi'')\det a_{ij}.
\end{align}
Some useful identities that we have used are: $\alpha_i = \sgn(A) y_i/\alpha$ so that $\alpha_i\alpha^i = \sgn(A)$ and $\alpha_i b^i = \sgn(A) s$. We conclude that
\begin{align}
\boxed{
\det g_{ij} = \phi^{n+1}(\phi-s\phi')^{n-2}(\phi-s\phi' +  (\sgn(A) b^2-s^2)\phi'')\det a_{ij}.}
\end{align}

In the case that $\alpha$ is positive definite, sign$(A)= 1$ everywhere, so the formula reduces to the standard result (see e.g. \cite{ChernShen_RiemannFinsler}).



\bibliography{GeneralBib}

\end{document}